\newif\ifC
\Cfalse  % toggle here

\ifC
    \documentclass[a4paper,UKenglish,cleveref, autoref, thm-restate]{lipics-v2021}
\else
    \documentclass[11pt]{article}
\fi

\ifC
\else
\fi

\ifC
\author{Avi Kadria}
{Department of Computer Science, Bar-Ilan University, Ramat Gan 5290002, Israel}
{avi.kadria3@gmail.com}
{https://orcid.org/0000-0001-8449-3284}
{}

\author{Liam Roditty}
{Department of Computer Science, Bar-Ilan University, Ramat Gan 5290002, Israel}
{liam.roditty@biu.ac.il}
{https://orcid.org/0000-0002-5289-198X}
{Supported in part by BSF grants 2016365 and 2020356.}

\author{Virginia Vassilevska Williams}
{Department of Electrical Engineering and Computer Science and CSAIL, MIT, Cambridge, MA, USA}
{virgi@mit.edu}
{https://orcid.org/0000-0003-4844-2863}
{Supported in part by NSF CAREER Award 1651838, NSF Grants CCF-1909429 and CCF-2129139, BSF grants 2016365 and 2020356, a Google Research Fellowship, and a Sloan Research Fellowship.}

\authorrunning{A. Kadria, L. Roditty, and V. Vassilevska Williams}
\Copyright{Avi Kadria, Liam Roditty, and Virginia {Vassilevska Williams}}

\ccsdesc[500]{Theory of computation~Approximation algorithms analysis}
\ccsdesc[500]{Theory of computation~Shortest paths}
\ccsdesc[500]{Theory of computation~Complexity classes}
\ccsdesc[300]{Mathematics of computing~Graph algorithms}

\keywords{Fine-grained complexity, Graph algorithms, shortest cycle, girth approximations} %TODO mandatory; please add comma-separated list of keywords

\relatedversion{}\relatedversiondetails{Full Version}{https://arxiv.org/abs/2607.00938}

\EventEditors{Philip Bille, Seth Pettie, and Sabine Storandt}
\EventNoEds{3}
\EventLongTitle{34th Annual European Symposium on Algorithms (ESA 2026)}
\EventShortTitle{ESA 2026}
\EventAcronym{ESA}
\EventYear{2026}
\EventDate{August 31--September 4, 2026}
\EventLocation{L'Aquila, Italy}
\EventLogo{}
\SeriesVolume{388}
\ArticleNo{80}
\else
\date{}
\author{Avi Kadria\thanks{Department of Computer Science, Bar Ilan University, Ramat Gan 5290002, Israel. E-mail {\tt avi.kadria3@gmail.com}.} \and Liam Roditty\thanks{Department of Computer Science, Bar Ilan University, Ramat Gan 5290002, Israel. E-mail {\tt liam.roditty@biu.ac.il}. Supported in part by BSF grants 2016365 and 2020356.} \and Virginia Vassilevska Williams\thanks{Department of Electrical Engineering and Computer Science and CSAIL, MIT, Cambridge, MA, USA. E-mail {\tt virgi@mit.edu}. Supported in part by NSF CAREER Award 1651838, NSF Grants  CCF-1909429 and CCF- 2129139, BSF grants 2016365 and 2020356, a Google Research Fellowship and a Sloan Research Fellowship.}}
\fi

\ifC
\else
\usepackage[margin=1in]{geometry}
\fi
\usepackage{setspace}
\usepackage{multirow}
\usepackage{makecell}
\usepackage[bottom]{footmisc}

\usepackage{epigraph}
\usepackage{graphicx}
\usepackage{tikz}
\usetikzlibrary{arrows,automata,patterns}
\ifC\else\usepackage{enumitem}\fi
\usepackage{hyperref}
\usepackage{lipsum}
\usepackage{verbatim}

\usepackage[linesnumbered,ruled,vlined,boxed,algo2e]{algorithm2e}
\usepackage{nicefrac}
\usepackage{amsmath}
\usepackage{float}
\usepackage{mathtools}
\usepackage{amsthm}
\ifC\else
\usepackage{fullpage}
\fi
\usepackage{amsfonts}
\usepackage{clipboard}
\ifC
\else
\usepackage[nameinlink]{cleveref}
\newtheorem{theorem}{Theorem}

\newtheorem{lemma}{Lemma}

\newtheorem{claim}{Claim}[lemma]

\newtheorem{corollary}[theorem]{Corollary}

\fi
\newtheorem{cclaim}{Claim}[lemma]

\newtheorem{subclaim}{Claim}[cclaim]

\newtheorem{problem}{Problem}

\newtheorem{definition}[lemma]{Definition}

\newtheorem{hypothesis}[lemma]{Hypothesis}

\newcommand{\Null}{\mathtt{null}}

\Crefname{enumi}{(item)}{(items)}

\newcommand{\Reminder}[1]{

\vspace{0.5em}
\noindent\textbf{Reminder of~\autoref{#1}.} \textit{\Paste{#1}}
\vspace{0.5em}

}

\usepackage{xcolor}
\definecolor{DarkGreen}{RGB}{1,50,32}
\usepackage{pgfplots}
\pgfplotsset{compat=1.18}

\usepackage{silence}
\usepackage{chngcntr}
\counterwithin*{equation}{section}
\counterwithin*{equation}{subsection}

\begin{document}
\ActivateWarningFilters[pdftoc]
% General
\newcommand{\defeq}{:=}
%\newcommand{\eps}{\varepsilon}

% Comments
\newcommand{\liam}[1]{{\color{red} \textbf{Liam}: #1}} 
\newcommand{\avi}[1]{{\color{purple} \textbf{Avi}: #1}} 
\newcommand{\virgi}[1]{{\color{blue} \textbf{Virgi}: #1}} 
\newcommand{\new}[1]{{\color{blue} #1}}

\newcommand{\blue}[1]{{\color{blue}#1}}
\newcommand{\eps}{\varepsilon}
% Pseudocode
\newcommand{\ReturnCode}{\textbf{return}}
\newcommand{\codestyle}[1]{\texttt{#1}}
\newcommand{\algorithm}[1]{\expandafter\newcommand\csname #1\endcsname{\mbox{\normalfont{\codestyle{#1}}}\xspace}}
\algorithm{AllEdgeTriangleOrCycleListing}
\algorithm{BallOrCycleListing}
\algorithm{DenseTriangleOrListing}
\algorithm{DegenerateOrCycle}
\algorithm{BallOrCycle}
\newcommand{\CC}{\mathcal{C}}
\newcommand{\VV}{\mathcal{V}}

\newcommand{\Initialize}{\mbox{\codestyle{Initialize}}}
\newcommand{\Cycle}{\mbox{\codestyle{Cycle}}}
\newcommand{\ADO}{\mbox{\codestyle{ADO}}}
\newcommand{\hADO}{\mbox{\codestyle{hADO}}}
\newcommand{\Query}{\mbox{\codestyle{.Query}}}
\newcommand{\ADOQuery}{\mbox{\codestyle{ADO.Query}}}
\newcommand{\ConstructADO}{\codestyle{ConstructADO}}
\newcommand{\FastAPSP}{\codestyle{FastAPSP}}
\newcommand{\Construct}{\mbox{\codestyle{Construct}}}
\newcommand{\Spanner}{\codestyle{Spanner}}
\newcommand{\RT}{\mbox{\codestyle{RT}}}
\renewcommand{\L}{\mbox{\codestyle{L}}}
\newcommand{\CycleOdd}{\codestyle{CycleOdd}}
\newcommand{\ClusterOrCycleBounded}{\codestyle{ClusterOrCycleBounded}}
\newcommand{\ClusterOrCycle}{\codestyle{ClusterOrCycle}}
\newcommand{\SimpleCycle}{\codestyle{SimpleCycle}}
\newcommand{\Next}{\codestyle{Next}}
\newcommand{\Sample}{\codestyle{Sample}}
\newcommand{\Dijkstra}{\codestyle{Dijkstra}}
\newcommand{\Preprocess}{\codestyle{Preprocess}}
\newcommand{\HashTable}{\codestyle{HashTable}}
\newcommand{\Heap}{\codestyle{Heap}}
\newcommand{\RelaxNext}{\codestyle{RelaxNext}}

\newcommand{\PreprocessGraph}{\codestyle{Initialize}}
\newcommand{\Route}{\codestyle{Route}}
\newcommand{\TreeRoute}{\codestyle{TreeRoute}}
\newcommand{\N}{\mathbb{N}}
\newcommand{\MinCycle}{\codestyle{MinCycle}}
\newcommand{\BoundedBFS}{\codestyle{BoundedBFS}}

\newcommand{\Ball}{\codestyle{Ball}}
\newcommand{\DistanceOracle}{\codestyle{TZ-DistanceOracle}}
\newcommand{\SparseOrCycle}{\codestyle{SparseOrCycle}}
\newcommand{\Intersection}{\codestyle{Intersection}}
\newcommand{\CycleAdditive}{\codestyle{CycleAdditive}}
\newcommand{\GenerateSi}{\codestyle{ComputeS}}

\newcommand{\ExtractMin}{\mathtt{ExtractMin}}
\newcommand{\Insert}{\mathtt{Insert}}
\newcommand{\codeNull}{\codestyle{null}}
\newcommand{\codeYes}{\codestyle{Yes}}
\newcommand{\codeNo}{\codestyle{No}}
\newcommand{\codeAnd}{~ \mathrm{and} ~}
\newcommand{\codeOr}{~ \mathrm{or} ~}
\newcommand{\wt}{\ell}
\newcommand{\Cl}{CL}
\newcommand{\CL}{CL}
\newcommand{\cl}{c\ell}

\newcommand{\LE}{\;\le\;}
\newcommand{\EQ}{\;=\;}
\newcommand{\GE}{\;\ge\;}
\newcommand{\Ot}{\tilde{O}}
\newcommand{\stactri}{\stackrel\triangle}

\newcommand{\AllEdgeCycleDetection}{\textsc{AllEdgeCycleDetection}\xspace}
\newcommand{\AllNodeCycleListing}{\textsc{AllNodesCycleListing}\xspace}
\newcommand{\AllEdgeCycleListing}{\textsc{AllEdgeCycleListing}\xspace}
\newcommand{\AllEdgeCycle}{\textsc{AllEdgeCycle}\xspace}
\newcommand{\AllNodeCycle}{\textsc{AllNodeCycle}\xspace}
\newcommand{\AllEdgeTriangleListing}{\textsc{AllEdgeTriangleListing}\xspace}
\newcommand{\AllEdgeTriangleDetection}{\textsc{AllEdgeTriangleDetection}\xspace}
\newcommand{\AllEdgeTriangle}{\textsc{AllEdgeTriangle}\xspace}
\newcommand{\AllEdgeTriangleOrListing}{\textsc{AllEdgeTriangleOrListing}\xspace}
\newcommand{\ThreeSum}{$3$\textsc{Sum}\xspace}
\newcommand{\StrongerThreeSum}{\textsc{Stronger}\ThreeSum}
\newcommand{\ZeroTriangle}{\textsc{ZeroTriangle}\xspace}
\newcommand{\APSP}{\textsc{APSP}\xspace}
\newcommand{\MinWeightCycle}{\textsc{MinimumWeightCycle}\xspace}
\newcommand{\MinWeightTriangle}{\textsc{MinimumWeightTriangle}\xspace}
\newcommand{\AllEdgeUniqueCycle}{\textsc{AllEdgeUniqueCycle}\xspace}
\newcommand{\Yes}{\codestyle{Yes}}
\newcommand{\No}{\codestyle{No}}

\newcommand{\EE}{\mathbb{E}}
\newcommand{\RR}{\mathbb{R}}

\ifC
\renewcommand{\paragraph}[1]{\textbf{#1}}
\else
\fi

\DeclarePairedDelimiter{\ceil}{\lceil}{\rceil}
\DeclarePairedDelimiter{\floor}{\lfloor}{\rfloor}
\DeclarePairedDelimiter{\pair}{\langle}{\rangle}

% \emergencystretch=3em % absorb minor overfull \hbox warnings in tight paragraphs
\title{Tighter bounds for weighted and unweighted shortest cycle approximation}

\maketitle

\thispagestyle{empty}

\begin{abstract}
We study the problem of approximating the length of a shortest cycle in a given graph, known as the girth of the graph. The state-of-the-art approximation algorithms for unweighted graphs by Kadria et al. [SODA'22] and Roditty and Trabelsi [arXiv'25] achieve the following trade-off: for every integer $k\geq 2$, there is an $\tilde{O}(n^{1+2/k})$ time algorithm that achieves a $(2k/3)$-approximation for the girth in unweighted $n$-node graphs. 
The first result of this paper is to achieve the same trade-off for $m$-edge, $n$-node graphs with non-negative real edge weights: a $2k/3$-approximation algorithm running in $\tilde{O}(m+n^{1+2/k})$ time. The dependence on $m$ is unavoidable in weighted graphs.
Our result improves on the work of Kadria et al.~[SODA'23] and Ducoffe [ICALP'19 and SIDMA'21], who were only able to achieve such a trade-off for some values of $k$.
We also prove new fine-grained lower bounds for girth approximation and related problems in unweighted graphs.
\end{abstract}
\clearpage
\pagenumbering{arabic} 
\newpage

\section{Introduction}
% The length of a shortest cycle, known as the \emph{girth} of the graph,
The length of a shortest cycle in a graph, known as the \emph{girth} and denoted by $g$, is a key parameter often used to shed light on the structure of graph problems (e.g., \cite{lazebnik1997structure,osthus2001almost,hoppen2016properties}). 
The problem of computing a shortest cycle and girth in an undirected graph is a fundamental problem studied extensively for decades, both in unweighted graphs (e.g., \cite{ItaiR78, AYZ97, YusterZ97, LingasL09, DBLP:conf/soda/RodittyW12, kadria2022algorithmic}), and weighted graphs (e.g., \cite{LingasL09,DBLP:conf/focs/RodittyW11,RodittyT13,DBLP:journals/siamdm/Ducoffe21,kadria2022algorithmic,kadria2023improved}).

Computing the exact value of the girth is computationally expensive: in weighted graphs, it is known to be equivalent to the All-Pairs Shortest Paths (APSP) problem, and in unweighted graphs, it is known that any fast algorithm requires Boolean Matrix Multiplication (BMM) \cite{WilliamsW18}. Because exact computation is deemed prohibitive, fast approximation algorithms have been extensively studied.
A (multiplicative) $c$-approximation algorithm outputs a cycle of length $\hat{g}\leq c\cdot g$ whenever $g$ is the girth of the underlying graph.
The factor $c$, which is the largest ratio $\frac{\hat{g}}{g}$ that the algorithm achieves, is called the \textit{stretch} of the algorithm.
%The stretch of an algorithm is defined as $\frac{\hat{g}}{g}$.

In unweighted graphs, Itai and Rodeh \cite{ItaiR78} were the first to consider girth approximation, and presented an $O(n^2)$ time algorithm that returns $g \le \hat{g} \le 2\ceil{g/2}$. Lingas and Lundell \cite{LingasL09} presented an $\Ot(n^{1.5})$ time algorithm\footnote{$\Ot$ omits poly-logarithmic factors.} that returns $g \le \hat{g} \le 4\ceil{g/2}$. Later, Kadria, Roditty, Sidford, Vassilevska Williams, and Zwick~\cite{kadria2022algorithmic} generalized these results and obtained the following trade-off. For any \textbf{even} integer $k \geq 2$, there is an $O(n^{1+2/k})$ time algorithm that returns $g \le \hat{g} \le k \cdot \lceil g/2 \rceil $. 
Recently, Roditty and Trabelsi \cite{roditty2025new} obtained the same trade-off for odd values of $k$.
The corresponding stretch is at most $\frac{k}{2}$ when $g$ is even, and at most $\frac{k}{2}\cdot(1+\frac{1}{g})$ when $g$ is odd. 
If $g=3$, that is, there is a triangle in the graph, then $\frac{k}{2}\cdot(1+\frac{1}{g})=2k/3$, which is the stretch in the worst case.
We can summarize the state of the art for \textit{unweighted} graphs as:

% %\virgi{does this footnote need to change? e.g. doesn't their result have additive error? does it need to be a footnote? does it supersede the result for $g=3,4$ mentioned below?}   \footnote{
% %    Recently, Roditty and Trabelsi \cite{roditty2025new} obtained the same result also for odd values of $k$.  
% %}

% In particular, for the special case that $g=3$ or $g=4$, Kadria~\etal \cite{kadria2022algorithmic} developed an algorithm that, for every integer (odd or even) $k\geq 2$,
% in $O(n^{1 + \frac{2}{k}})$ time returns a cycle of length at most $2k$. 
% The corresponding stretch in the special case that $g=3$ is at most $2k/3$, for \textbf{any} integer $k\geq 2$ (and the stretch for $g= 4$ is only better). 

% Recently, Roditty and Trabelsi \cite{roditty2025new} obtained a set of new approximation results for unweighted graphs. In particular, they showed that for every integer $\ell\geq 2$ and any $\eps>0$, there is an $\tilde{O}(\ell n^{1+1/(\ell-\eps)})$ time algorithm that returns a cycle of length at most $2\ell\ceil{g/2}-2\lfloor\eps\ceil{g/2}\rfloor$. For any odd integer $k\geq 3$, we can set $\ell=(k+1)/2$ and $\eps=1/2$ to obtain running time $\tilde{O}(n^{1+2/k})$ and stretch $k/2+(k+2)/(2g)$ for odd $g$ (and better for even $g$). This stretch is $\leq 2k/3$ for every $g\geq 5$ and $k\geq 3$. Combined with the aforementioned results of \cite{kadria2022algorithmic} 

\begin{center}
For every integer $k\geq 2$, there is an $\tilde{O}(n^{1 + \frac{2}{k}})$ time $2k/3$-approximation algorithm for the girth of unweighted graphs.
\end{center}

% \avi{Rev:Is the bound of 2k/3 tight? Is there a good reason to suspect that the bound is tight, even if we cannot prove it?
% Ans: I have an intuition (but I'm not sure if it is worth writing - something about that if triangle detection in girth conjecture graphs requires $n\cdot (m/n)^2$ then this is tight).
% }

A central and extensively studied question in graph algorithms is whether the running time - approximation trade-off established for unweighted graphs can be matched in the case of weighted undirected graphs with non-negative real edge weights. Therefore, the following problem is natural.

\begin{problem}\label{P-1}
Is it possible, for every integer $k \ge 2$, to obtain an $\tilde{O}(m + n^{1 + \frac{2}{k}})$ time algorithm with $\frac{2k}{3}$-stretch in weighted undirected graphs with non-negative real edge weights?
\end{problem}

(We allow an additive $m$ in the running time above since in weighted graphs one needs to read the input even for an approximation\footnote{Take a complete graph with infinite edge weights and reduce the weight of the edges in some triangle to $1$; if an algorithm does not find any of the $1$ weight edges, it can return only girth of $3\cdot \infty$, and finding $3$ needles in a stack of $O(n^2)$ size requires $\Omega(n^2)$ time.}.)

Roditty and Tov~\cite{RodittyT13} almost solved \Cref{P-1} for $k=2$ and presented an $\Ot(\frac{1}{\eps}n^{2})$ time algorithm with $(4/3+\eps)$-stretch. Later, Ducoffe~\cite{DBLP:journals/siamdm/Ducoffe21} almost solved \Cref{P-1} for $k=3$ and presented an $\Ot(\frac{1}{\eps}(m+n^{5/3}))$ time algorithm with $(2+\eps)$-stretch.
\cite{kadria2023improved} solved \Cref{P-1} for every \textbf{even} integer $k\geq 2$, and presented an $\tilde{O}(m + n^{1 + \frac{2}{k}})$ time algorithm with $\frac{2k}{3}$-stretch.

Therefore, to solve \Cref{P-1} for every $k\geq 2$ and establish that the time/stretch trade-off in weighted graphs matches that of unweighted graphs, it remains to address \Cref{P-1} for \textbf{odd} values of $k \ge 3$.
In this paper, we completely solve \Cref{P-1} and prove the following.

\begin{theorem}\label{T-main-odd}\Copy{T-main-odd}{
     Let $G=(V, E, \ell)$ be a weighted graph, where $\ell: E\rightarrow \RR_{\ge 0}$, and let $k\ge 3$ be an \textbf{odd} integer. There exists an algorithm that runs in $\Ot(m+n^{1+\frac{2}{k}})$ time and returns an estimation $\hat{g}$ such that $g\le \hat{g} \le \frac{2k}{3}g$.
     }
\end{theorem}

We note that our result improves the known trade-off for real weighted graphs for all odd values of $k\ge 3$. In the special case of $k=3$, Ducoffe's algorithm~\cite{DBLP:journals/siamdm/Ducoffe21} achieves a $2=2\cdot (3/3)$-approximation only for graphs with polynomially bounded integer weights. For real weights, the stretch of Ducoffe's algorithm is $(2+\eps)$ for an arbitrarily small constant $\eps>0$, whereas ours is truly $2$.

Among the tools we use to prove the theorem are: approximate distance oracles \cite{DBLP:journals/jacm/ThorupZ05}, Spira's single-source shortest paths algorithm \cite{Spira73}, ideas from Kadria et al.~\cite{kadria2023improved}, and a generalization of the ideas of \cite{DBLP:journals/siamdm/Ducoffe21}.

After matching the trade-off of weighted and unweighted graphs, we consider {\bf fine-grained lower bounds} for girth approximation and related problems in unweighted graphs. 

\paragraph{Fine-grained Lower Bound for Girth Approximation.}
Triangle detection
%, which can be thought of $4/3-\eps$ stretch for girth approximation, 
is one of the most important special cases of the shortest cycle problem. It plays a fundamental role in algorithm design and complexity theory and also has practical applications (e.g. \cite[Chapter~3]{kleinbergbook}). Any $(4/3-\eps)$-approximation algorithm for the girth (for $\eps>0$) can distinguish between girth $\geq 4$ and the existence of a triangle, and thus triangle detection can be viewed as achieving stretch $(4/3-\eps)$ for girth approximation.

The first conditional lower bounds for triangle detection come from the aforementioned equivalence results of
Vassilevska W. and Williams~\cite{WilliamsW18}. 
These imply that any $O(n^{3-\eps})$ time algorithm (for $\eps>0$) for triangle detection would also imply an $O(n^{3-\eps'})$ time algorithm (for $\eps'>0$) for BMM, and hence under the popular BMM Hypothesis, ``combinatorial'' techniques cannot achieve $O(n^{3-\eps})$ time for triangle detection and hence for $(4/3-\eps)$-girth approximation. Thus, practical fast algorithms for $(4/3-\eps)$-girth approximation are considered out of reach. 

For sparse graphs, \cite{DBLP:conf/stoc/AbboudBKZ22} showed that triangle detection in an $m$ edge graph with no $4$ cycles requires $m^{1.1194-o(1)}$ time, assuming that triangle detection in $\sqrt n$-degree graphs requires $n^{2-o(1)}$ time. This also implies a conditional $m^{1.1194-o(1)}$ time lower bound for any $5/3-\eps$ girth approximation algorithm for $\eps>0$. While the hypothesis used by \cite{DBLP:conf/stoc/AbboudBKZ22} is somewhat non-standard, a weaker version of it was shown by \cite{jin2023removing} to be implied by the slightly more popular Strong \ThreeSum Hypothesis, giving evidence that the non-standard hypothesis could be true.
(\cite{jin2023removing} also based the hardness of the classical problem of triangle detection on the Strong \ThreeSum hypothesis.\footnote{Later, Chan and Xu \cite{ChanX24} improved the triangle detection lower bound to $m^{9/7-o(1)}$, under the so-called Strong Exact Triangle Hypothesis.})

%To obtain a lower bound for the detection variant, we take inspiration from \cite{jin2023removing}, who based the hardness of triangle detection 

The Strong \ThreeSum hypothesis studied by \cite{jin2023removing}  states that \ThreeSum on $n$ numbers from the universe $[\pm O(n^2)]$ requires $n^{2-o(1)}$ time (on a word-RAM).\footnote{Where $[\pm n]= \{i \in \mathbb{Z} \mid |i| \le n\}$}
%\footnote{Later, Chan and Xu \cite{ChanX24} improved the triangle detection lower bound to $m^{9/7-o(1)}$, under the Strong Exact Triangle Hypothesis.}
We consider the following \StrongerThreeSum hypothesis: On a word-RAM with $O(\log n)$ bit words, \ThreeSum on $n$ numbers from a (Sidon) set\footnote{A set $A$ of integers is a \emph{Sidon set} if there are no distinct $a,b,c,d\in A$ such that $a+b=c+d$.} $A\subseteq [\pm n^{2+o(1)}]$ with no non-trivial solution to $a+b=c+d$ requires $n^{2-o(1)}$ time.
% \virgi{Is this really $a+b+c+d=0$? I thought we were looking at Sidon sets so it should be $a+b=c+d$. }\virgi{Rewriting as Sidon set.}

This hypothesis seems much stronger than the Strong \ThreeSum hypothesis. However, it is still plausible. In fact, \cite{jin2023removing} explicitly ask (Open problem 5) whether their techniques for fine-grained reductions for \ThreeSum in Sidon Sets and Sidon Set verification can be improved so that the range of the integers does not increase by much and the problems are still hard for $n$ integers in the range $[\pm n^{2+\delta}]$ for small $\delta$. So far, all known techniques for removing additive structure in fine-grained reductions really blow up the size of the integers so that it is unclear whether one can prove the \StrongerThreeSum hypothesis with current techniques. The hypothesis could also be false due to the highly structured nature of the input, but proving that it is false would also seem to require interesting new techniques.
We show:
\begin{theorem}\label{T-Stronger-3sum-LB}\Copy{T-Stronger-3sum-LB}{
    Under the \StrongerThreeSum Hypothesis, there is no $O(n^{1.5-\eps})$ time algorithm for $C_{\le 4}$ detection and hence for $(5/3-\eps)$-girth approximation in graphs with maximum degree $n^{1/4}$.
    }
\end{theorem}

A brute-force algorithm can detect if a given $n$-node graph with maximum degree $n^{1/4}$ has a triangle or a $C_4$ in time $O(n^{1.5})$: for every vertex $u$, try all pairs of its neighbors $v,v'$ and check whether $(v,v')$ is an edge, thus checking if a triangle exists; if no triangle is found, the previous step actually lists all two-edge paths in the graph $v-u-v'$ and one can detect a $4$-cycle just by sorting\footnote{Sorting here can be done in linear time because the $n$ nodes can be uniquely labeled by the integers in $[n]$.} these by their end-points and checking for a collision. Thus our theorem shows that under the \StrongerThreeSum Hypothesis, the brute-force algorithm is essentially optimal.

Previously, the best known lower bound for $(5/3-\eps)$-girth approximation was by \cite{DBLP:conf/stoc/AbboudBKZ22}. It is instructive to compare the time lower bounds in terms of the number of edges. That of \cite{DBLP:conf/stoc/AbboudBKZ22} is $m^{1.1194-o(1)}$ (as mentioned earlier) and ours is $m^{6/5-o(1)}=m^{1.2-o(1)}$. Of course, they are under different hypotheses that may or may not be true.

\paragraph{Listing Cycles or All Edge Triangles.}
A relatively recent line of work concerns algorithms and conditional lower bounds for listing and enumerating cycles in graphs. Listing triangles has been studied extensively both from an algorithmic point of view (e.g. \cite{trianglelisting}) and from a fine-grained perspective \cite{williams2020monochromatic,DBLP:conf/soda/KopelowitzPP16,Patrascu10}. Listing and enumerating cycles of larger constant length is an even more recent topic of study \cite{VW25,JinWZ24,DBLP:conf/stoc/AbboudBKZ22,jin2023removing,DBLP:conf/stoc/AbboudBF23}.

Most known conditional lower bounds for cycle listing and enumeration go through the All-Edge Sparse Triangle problem: given an $m$-edge graph, determine for every edge $e$ whether $e$ appears in a triangle. In 2010, Patrascu \cite{Patrascu10} showed that under the \ThreeSum hypothesis, All-Edge Sparse Triangle requires $m^{4/3-o(1)}$ time. Since then various refinements of this problem are studied from a lower bounds perspective, most notably for solving All-Edge Sparse Triangle in graphs with a small number of cycles of length at most $k$ \cite{DBLP:conf/stoc/AbboudBKZ22,jin2023removing,DBLP:conf/stoc/AbboudBF23}.

We study the recent techniques for conditional lower bounds for listing cycles and discover that they can be used to prove hardness for an easier problem as well.

Consider the following \AllEdgeTriangleOrListing problem: for an integer $k\geq 4$, given a graph $G$ and an integer $t$,
%with $t$ cycles of length at most $k$, 
either solve the All-Edge Sparse Triangle in $G$, or list $t$ cycles of length at most $k$.

This problem is easier than both All-Edge Sparse Triangle and the problem of listing at most $t$ cycles of length at most $k$, as an algorithm for it can choose which problem to solve depending on the input graph. 

Listing $t$ cycles of length at most $k$ is an easier problem than listing $t$ cycles of length exactly $k$, the problem of recent interest as described earlier. For instance, there is a very simple $O(n^2)$ time algorithm for finding a single cycle of length at most $2k$ for any integer $k$ \footnote{Run BFS from every node up to $k$ levels and stop when the first cycle is closed. If no cycle is closed, each BFS only sees a tree and runs in $O(n)$ time.}, whereas an $O(n^2)$ time algorithm for finding a cycle of length exactly $2k$ is much more involved \cite{YusterZ97}. However, the problem of listing all cycles of length at most $k$ is at least as hard as that of listing all cycles of length $k$. So the two listing problems are tightly linked.

%When run with $t$ set to the number of cycles of length $\leq k$, \AllEdgeTriangleOrListing also solves a version of AESC approximation: for every edge in the graph that is contained in a cycle of length $\leq k$, 
%\AllEdgeTriangleOrListing provides a (at worst) $k/3$-approximation. So in a certain sense, \AllEdgeTriangleOrListing lies between $k/3$-approximating AESC and both All-Edge Sparse Triangle and Cycle Listing. 

%\avi{I'm not sure why \AllEdgeTriangleOrListing solves AESC, what if the number of $C_{\le k}$ is very large, in this case we are not guaranteed to report a cycle for every edge (but just to list $t$ of them, maybe we can say that this solves AESC in graphs with small number of $C_{\le k}$, but it actually solves the harder AE-Triangle in such graphs.}

By simple modifications of the known techniques, we show:
\begin{theorem}
Under the $3$SUM Hypothesis, there is no $O(m^{1+1/(k-1)-\eps}+t)$ time algorithm for \AllEdgeTriangleOrListing.
\end{theorem}

As this is an intermediary problem that has not been studied before, it may not be clear how close this lower bound is to the true complexity of the problem. We show that, up to constants in front of $k$, the lower bound is in fact tight.

\begin{theorem}
There is an $O(m^{1+1/(k+1)}+t)$ time algorithm that given an $m$-edge graph $G$ either returns $t$ cycles of length at most $2k$ in $G$ or solves the All-Edge Sparse Triangle problem in $G$.
\end{theorem}

This should be compared with the fastest known algorithm for listing $2k$ cycles in terms of $m$. It is believed that the ``right'' running time for listing $t$ $2k$-cycles should be $\tilde{O}(m^{2k/(k+1)}+t)$ \cite{VW25} as the best known running time for $2k$-cycle detection is $\tilde{O}(m^{2k/(k+1)})$ \cite{DahlgaardKS17}. This running time has been achieved for $k=2$ \cite{jin2023removing,DBLP:conf/stoc/AbboudBF23} and is a big open problem for $k\geq 3$ (though it has been almost achieved for $k=3$ \cite{VW25}). 
Our theorem above can be viewed as a step in the direction of understanding the complexity of listing cycles of length {\em at most} $k$ (as opposed to exactly $k$).

\section{Preliminaries}
Let $G = (V, E,\ell)$  be a weighted undirected graph, where $\ell: E \to (0,\infty)$ is a real \emph{length} function defined on its edges. Let $n=|V|$ and $m=|E|$. 
The graph is represented using an adjacency-list representation. We assume that the edges incident on a vertex~$u$ are sorted in a non-decreasing order of length. 
\footnote{
A priori, the algorithms of \cite{kadria2022algorithmic} work in $o(m)$ if we assume the edges are sorted. However, \cite{kadria2023improved} proved that even in this model, $\Omega(m)$ time is required to achieve approximation that is better than $4k$ stretch in $O(n^{1+2/k})$ time. Therefore to get the improved $8k/3$-stretch, even in this model our algorithm needs to have $\Omega(m)$ time.}
% (If not, this can be easily computed in $O(m\log n)$ time.)

For all $u,v\in V$, let $P(u,v)=\{u=x_1,\ldots,x_t=v\}$ be a shortest path between $u$ and $v$ in $G$, let $\pi(u,v)$ be the last edge on a shortest path from~$u$ to~$v$, and let $P_2(u,v)$ be a second-shortest path from $u$ to $v$ (i.e. $P_2(u,v)$ is a shortest path from $u$ to $v$ that differs from $P(u,v)$ by at least one edge).
Let $\ell(P(u,v))=\sum_{i=1}^{t-1}\ell(x_i,x_{i+1})$. Let the \emph{distance} $\delta(u,v)$ between $u$ and $v$ be $\ell(P(u,v))$.
% \footnote{When the graph~$G$ is clear from the context, which will almost always be the case, we write $\delta(u,v)$ and $P(u,v)$ instead of $\delta(u,v)$ and $P(u,v)$.} 
A sequence of vertices $C=(x_1,x_2, \ldots, x_t)$ is a cycle if $x_1=x_t$, and $(x_i, x_{i+1})\in E$, for every $1\le i \le t-1$. If $x_1$ is the only vertex that appears twice in $C$, then $C$ is a \emph{simple} cycle. Let $C$ be a simple cycle. We denote $\ell(C)=\sum_{i=1}^{t-1}\ell(x_i,x_{i+1})$.
The \emph{girth} $g$ of a graph $G=(V, E,\ell)$ is the length of a shortest simple cycle in~$G$. 
That is, $C$ is a shortest cycle if $\ell(C)=g$. Let $M(C)$ be an edge with the largest length in~$C$.
We remark that throughout the paper, we treat shortest paths $P(\cdot,\cdot)$ and cycles $C$ interchangeably as sets of edges and sets of vertices. For the rest of the paper, we also assume that the graph is connected and contains a simple cycle, as otherwise the problem is trivial.

$\Sample(S,p)$ is a procedure that returns a subset $S^*$ of $S$ where each element of $S$ is included in  $S^*$ independently with probability $p$. 
Given a set $S\subseteq V$, the induced subgraph $G[S]$ is defined as $(S,\{(u,v)\in E \mid u,v\in S\})$.
If $u\in V$ and $A\subseteq V$, we let $\delta(u,A)=\min_{v\in A} \delta(u,v)$ denote the \emph{distance from~$u$ to the set~$A$}. (If $A=\emptyset$, then $\delta(u,A)=+\infty$.)
Given a set $S$ we define $B_S(u)=\{v\in V \mid \delta(u,v) < \delta(u,S) \}$. 

\begin{lemma}[Nearest set computation \cite{DBLP:conf/soda/ChechikLRSTW14}]\label{L-Create_B_S}
    Let $S=\Sample(V,n^{-x})$. Computing $G[B_S(u)]$ for every $u\in V$ takes $O((m+n^{1+2x})\log{n})$ time.
    % Given a random set $S$ of size $n^{1-x}$,
\end{lemma}
% \avi{Do we want to insert the proof?}

We denote by $\MinCycle(u, G)$ an $O(m + n \log n)$ time procedure that runs a modification of Dijkstra's algorithm from $u$ in $G$ and computes both $P(u,v)$ and $P_2(u,v)$, for every $v\in V$.\footnote{This can be implemented by maintaining two priority queues during Dijkstra's algorithm-one for the shortest distances and another for the second-shortest distances.} $\MinCycle$ then computes $v=\arg\min_{v\in V}(\ell(P(u,v))+\ell(P_2(u,v)))$, removes the common prefix and suffix of $P(u,v)$ and $P_2(u,v)$, and returns the resulting simple cycle. Notice that if $u$ is on a shortest cycle, then $\ell(\MinCycle(u, G))=g$.

Following~\cite{kadria2023improved}, we define the distance from
a vertex~$u\in V$ to an edge $(v,w)\in E$ as follows: $\delta(u,(v,w))=\min\{\delta(u,v), \delta(u,w)\}+\ell(v,w)$. \footnote{Note that $\delta(u,(v,w))=\delta(u,\{v,w\})+\ell(v,w)$. Here $\{v,w\}$ is a set of two vertices.}
Let $u\in V$ and $r>0$. We define the \emph{ball graph} $G_r(u) = (V_r(u),E_r(u))$ of \emph{radius} $r$ around~$u$ as follows:
\begin{align*}
    V_r(u) \EQ \{ v\in V \mid \delta(u,v)\leq r\} 
    \text{ and }
    E_r(u) \EQ \{ e\in E \mid \delta(u,e)\leq r\} \;.
\end{align*}
% Note that $G_r(u)$ is not necessarily the same as $G[V_r(u)]$.
% For example, $G[V_r(u)]$  may include edges $(v,w)$ such that $\delta(u,(v,w))>r$, whereas such edges are excluded from $E_r(u)$.
We let $G_{<r}(u)=(V_{<r}(u),E_{<r}(u))$ denote the \emph{open} ball graph of radius~$r$ around~$u$. The definitions of $V_{<r}(u)$ and $E_{<r}(u)$ are identical to those of $V_{r}(u)$ and $E_{r}(u)$ with the weak inequalities $\delta(u,v)\leq r$ and $\delta(u,e)\leq r$ replaced by strict inequalities.

The following general and simple lemma from \cite{kadria2023improved} is used in the correctness proof of our algorithm, and therefore, we provide its proof here for completeness.
\begin{lemma}[\cite{kadria2023improved}] \label{lem:main}
Let  $G = (V, E,\ell)$ be a weighted undirected graph, $C$ a cycle in~$G$, $u\in V$, and $r>0$. If $V_r(u)\cap C\ne \emptyset$, then $C\subseteq G_{r+\frac{1}{2}(\ell(C)+M(C))}(u)$.
\end{lemma}
\begin{proof} 
Let $v\in V_r(u)\cap C$. By definition $\delta(u,v)\le r$. 
Let $(x,y)\in C$. Assume, without loss of generality, that $\delta(v,x)\le \delta(v,y)$. As $\delta(v,x)+\ell(x,y)+\delta(v,y)\le \ell(C)$, we get that $\delta(v,x)\le \frac12(\ell(C)-\ell(x,y))$. Thus
\begin{align*}
    \delta(u,(x,y)) &\LE \delta(u,v)+\delta(v,x)+\ell(x,y) \\
    &\LE r + \frac12(\ell(C)-\ell(x,y)) + \ell(x,y) \\
    &\EQ r+\frac12(\ell(C)+\ell(x,y)) 
     \LE r+\frac12(\ell(C)+M(C)),
\end{align*}
where the last inequality follows from the fact that $(x,y)\in C$, and therefore $\ell(x,y)\le M(C)$.
Thus, $(x,y)\in E_{r+\frac{\ell(C)+M(C)}{2}}(u)$, for every $(x,y)\in C$ and therefore $C\subseteq G_{r+\frac{\ell(C)+M(C)}{2}}(u)$, as required.
\end{proof}

Let $V=A_0\supseteq A_1 \supseteq A_2 \supseteq \dots \supseteq A_k = \emptyset$ be a hierarchy of vertex sets, where $k \ge 1$. Following \cite{DBLP:journals/jacm/ThorupZ05} we let $B_i(u)=\{w\in A_i\mid \delta(u,w) < \delta(A_{i+1},u)\}$.

Following \cite{kadria2023improved} we define the \emph{cluster graph} of $u\in A_i\setminus A_{i+1}$ in~$G$ to be the graph $\Cl(u)=(\Cl_V(u),\Cl_E(u))$, where 
\begin{align*}
    \Cl_V(u) &\EQ \{ v\in V \mid \delta(u,v)<\delta(A_{i+1}, v) \} \;, \\
    \Cl_E(u) &\EQ \{ (v,w)\in E \mid \delta(u,v) + \ell(v,w) < \delta(A_{i+1}, w) \} \;.
\end{align*}
The cluster graph can be viewed as an extension of the clusters of Thorup and Zwick\cite{DBLP:journals/jacm/ThorupZ05} to a subgraph rather than a set of vertices. Notice that in the definition of $\Cl_E(u)$, we have that $v,w$ have asymmetrical roles, and it might be that $v,w$ would satisfy the condition, whereas $w,v$ would not.

For any $u\in V$ and $0\le i<k$, we let $p_i(u)=\arg \min_{v\in A_i} \delta(u,v)$, i.e., $p_i(u)$ is a vertex of~$A_i$ closest to~$u$ (ties are broken lexicographically). 
\cite{kadria2023improved} proved the following property on cluster graphs.
\begin{lemma}[\cite{kadria2023improved}]\label{L-shortest-path-cluster}
Let $u\in A_i\setminus A_{i+1}$. If $v\in \Cl_V(u)$ then $P(u,v)\subseteq \Cl(u)$.
\end{lemma}
% \begin{proof}
% Let $x$ be a vertex on a shortest path~$P$ from~$u$ to~$v$. Assume, for contradiction, that $x\notin\Cl_V(u)$. Let $w=p_{i+1}(x)\in A_{i+1}$. Then, $\delta(w,x)=\delta(x,A_{i+1})\le\delta(u,x)$. It follows that $\delta(w,v)\le \delta(w,x)+\delta(x,v)\le\delta(u,x)+\delta(x,v)=\delta(u,v)$, contradicting the claim that $v\in \Cl_V(u)$. The proof for the edges on~$P$ is similar.
% \end{proof}

Clusters have especially nice properties when the hierarchy $V=A_0\supseteq A_1 \supseteq A_2 \supseteq \dots \supseteq A_k = \emptyset$ is obtained using random sampling. Lemma~\ref{L-cluster-size} gives one such property and is proven in \cite{DBLP:journals/jacm/ThorupZ05} using a simple probabilistic argument. 

\begin{lemma}[\cite{DBLP:journals/jacm/ThorupZ05}]\label{L-cluster-size}
If $A_{i+1}$, for $i=0,1,\ldots,k-2$, is obtained by including each vertex of $A_i$ independently with probability $n^{-1/k}$, then $\EE[\sum_{u\in V}|\Cl_V(u)|]=O(kn^{1+1/k})$.
\end{lemma}

The main tool from ~\cite{kadria2023improved} that we need is algorithm $\ClusterOrCycle$ (see \Cref{AP-OldToold} for the pseudocode), which either constructs a cluster or reports a short cycle.
% using the vertex hierarchy $V=A_0\supseteq A_1 \supseteq A_2 \supseteq \dots \supseteq A_k = \emptyset$.
In~\cite{kadria2023improved} the vertex hierarchy $V=A_0\supseteq A_1 \supseteq A_2 \supseteq \dots \supseteq A_k = \emptyset$, where $k \ge 1$, is initialized using the procedure $\Initialize$ (see \Cref{AP-OldToold} for the pseudocode). The value of $\delta(u,A_i)$ and $p_i(u)$ are computed as well, for every $u\in V$ and $i\in [k]$. In the $\Initialize$ algorithm, $\Preprocess$ is also called to allow efficient access to edges in  $\ClusterOrCycle$. 
The properties of $\ClusterOrCycle$ are summarized in the following lemma:

\begin{lemma}[Lemma 5 in \cite{kadria2023improved}]\label{L-ClusterOrCycle}\Copy{L-ClusterOrCycle}{Let $G=(V,E,\ell)$ be a  weighted undirected graph on which we run procedure $\Initialize$ and let $u\in V$. If $\Cl(u)$ contains a cycle, let $r>0$ be the smallest number such that $\Cl(u)\cap G_r(u)$ contains a cycle, then $\ClusterOrCycle(u)$ returns a description of a cycle of length at most $2r$. Furthermore, it returns the shortest paths from $u$ to all vertices of $\Cl(u)\cap G_{<r}(u)$. Otherwise, if $\Cl(u)$ is a tree, then $\ClusterOrCycle(u)$ returns all the shortest paths from $u$ to $\Cl(u)$. $\ClusterOrCycle(u)$ can be implemented in $O(|\Cl_V(u)|\log n)$ time.
}
\end{lemma}

\section{Technical overview}
\paragraph{Algorithm Overview.} Our techniques build upon the distance oracle of Thorup and Zwick\\~\cite{DBLP:journals/jacm/ThorupZ05} and the algorithm of~\cite{kadria2023improved}, with an important and interesting difference.
While both the prior works \cite{DBLP:journals/jacm/ThorupZ05} and \cite{kadria2023improved} employ uniform sampling to construct a vertex hierarchy of $k$ levels, our algorithm instead uses {\em non-uniform sampling}. This non-uniform sampling enables the creation of what can be viewed as a ``half-level'' in the hierarchy, an idea we elaborate on below, which is crucial for achieving improved approximation guarantees for odd values of $k = 2\alpha + 1$.

% Roughly speaking, our algorithm works as follows.
Let $A_0=V$, and let $A_1\subseteq A_0$ be a random subset where each vertex of $A_0$ remains in $A_1$ with probability $n^{-1/k}$. Then, for every $2 \le i\le \alpha$, $A_i\subseteq A_{i-1}$ is a random subset where each vertex of $A_{i-1}$ remains in $A_i$ with probability $n^{-2/k}$, and let $A_{\alpha+1}=\emptyset$. 

Let $C$ be a shortest cycle of length $g$.
Since $A_1$ is sampled with probability $n^{-1/k}$, the expected size of $B_0(u)$ is $O(n^{1/k})$ (whereas $|B_i(u)| = O(n^{2/k})$,  for every $i > 0$). Our algorithm combines ideas from Ducoffe~\cite{DBLP:journals/siamdm/Ducoffe21} and Kadria, Roditty, Sidford, Vassilevska Williams, and Zwick\cite{kadria2023improved} with the fact that $B_0(u)$ is relatively small. 
This allows us to efficiently incorporate the ``half-level'' $A_1$ together with $A_2$ and to either find a short cycle satisfying $\hat{g} \le 2g$, or to guarantee that the cycle $C$ has a vertex sufficiently close to $A_2$, specifically $\delta(C, A_2) \le 1.5 \cdot \tfrac{2}{3}g = g$, where $\delta(C, A_2) = \min\{\delta(c, a) \mid  \pair{c, a} \in C \times A_2\}$. 

Then, by utilizing ideas from \cite{kadria2023improved} for the rest of the levels we show that either a cycle of length $\hat{g} \le 2\delta(A_{i-1}, C)+4g/3$ is found or $\delta(A_i, C) \le \delta(A_{i-1}, C)+\frac{2g}{3}$. 
Combined with the fact that $\delta(A_2,C) \le g$ we show that $\hat{g}\leq \frac{2k}{3}g$. 

Roughly speaking, our algorithm works as follows. 
Let $\hat{g}$ be the girth approximation, and let $\hat{C}$ be the cycle of length $\hat{g}$ found.
In Phase 1, we call $\MinCycle(G[B_0(u)])$, and update $\hat{g}$ and $\hat{C}$ accordingly.
Then, in Phase 2, for every $u\in A_1$ we call $\ClusterOrCycle(u)$ and update $\hat{g}$ and $\hat{C}$ accordingly, and finally, in Phase 3, we iterate over every edge $(u,v)\in E$, and check whether a short cycle can be closed using the edge $(u,v)$ and a precomputed distance.

The correctness proof is divided into two cases: the case that $M(C)\le g/3$, and the case that $M(C)\ge g/3$, where $M(C)$ is the maximum edge weight in $C$. (See \Cref{L-Corr-small,L-corr-big}.)

Let $(u,u')\in C$ be an edge such that $\ell(u,u')=M(C)$, and let\\ $h_i'(C)=\min(\delta(u, A_{i}), \delta(u', A_{i}))$, and let $h_i(C)=\min_{w\in A_i}(\delta(w,C))$ (see \Cref{fig:h_i} for an illustration). 

First, we show that in Phase 1 of our algorithm, either a short cycle is found or $h_2(C),h_2'(C)\le g$ (See \Cref{L-bound-h2}).

In the case that $M(C)\le g/3$, we then show that in Phase 2 of the algorithm, in every level we have that either $\hat{g} \leq 2h_i(C) + \frac{4}{3}g$ or $h_i(C) \le h_{i-1}(C)+2g/3$ (See \Cref{C-bound-hi}).

Since $A_{\alpha+1}=\emptyset$ and therefore $h_{\alpha+1}(C)=\infty$, by applying the above inequality $(\alpha-2)$ times, we get that:
$$\hat{g} \le 2\left(h_2(C) + \frac{2}{3}(\alpha - 2)g\right) + \frac{4}{3}g \stackrel{h_2(C)\le g}\le 2\left(g + \frac{2}{3}(\alpha - 2)g\right) + \frac{4}{3}g 
= \frac{2}{3}(2\alpha + 1)g,$$
as required.

Similarly, in the case that $M(C) > g/3$, we show that either a short cycle is found in Phase 2 or 3, or that $h_i'(C) \le h_{i-1}'(C)+2g/3$, and using the above inequalities for $h_i'(C)$ we have that $\hat{g} \le \frac{2}{3}(2\alpha + 1)g$.

\paragraph{Conditional Lower Bounds for $C_{\le 4}$ detection.}
We consider the $C_{\le 4}$  detection problem, which asks whether a graph contains a triangle or a cycle of length $4$.
Here we modify a reduction by~\cite{jin2023removing} from \textsc{Strong}\ThreeSum to triangle detection. We show that if one instead starts from the \textsc{Stronger}\ThreeSum Hypothesis, where the input \ThreeSum instance is assumed to both be a Sidon Set and have small numbers, then one can create polylogarithmically more instances of triangle detection that now contain no $4$-cycles and whp the \ThreeSum instance has a solution if and only if one of these instances contains a triangle. Hence, any efficient $C_{\le 4}$ detection algorithm can be used to solve the original \ThreeSum instance efficiently.

\paragraph{Upper and lower bound for \AllEdgeTriangleOrListing}
We also consider the \\$(k,t)$-\AllEdgeTriangleOrListing problem, which asks to either solve the \AllEdgeTriangle or list \textit{at least} $t$ cycles of length at most $k$.
To show a reduction from \ThreeSum to \\$(k,t)$-\AllEdgeTriangleOrListing, we consider the problem of \AllEdgeTriangle in sparse graphs with a bounded number of short cycles that was proven by \cite{jin2023removing} to require $n^{2-o(1)}$ time.
We then leverage a sampling technique used by \cite{DBLP:conf/stoc/AbboudBKZ22, DBLP:conf/stoc/AbboudBF23, jin2023removing} to further reduce the number of cycles of length at most $k$ to a point where the number of cycles found is less than $t$, and therefore the \AllEdgeTriangleOrListing in fact solves the \AllEdgeTriangle in these subgraphs. Using this technique we obtain the $O(m^{1+\frac{1}{k-1}})$ lower bound for the \AllEdgeTriangleOrListing problem.

In addition, we provide an upper bound algorithm for the \AllEdgeTriangleOrListing \\problem. To achieve this, we extend the \DegenerateOrCycle and \BallOrCycle procedures \\of \cite{kadria2022algorithmic} to handle the more challenging task of \AllEdgeTriangleOrListing, rather than girth approximation. Using these extended tools, we develop an algorithm that solves the 
($2k,t$)-\AllEdgeTriangleOrListing problem in $O(m+\min(m^{1+\frac{1}{k+1}},n^{1+2/k})+tk)$ time.

\section{Weighted girth approximation algorithm}
In this section\footnote{This section has missing proofs due to the line limit, for the full version see \Cref{S-algo-full}.}, we prove the following theorem:
\Reminder{T-main-odd}

Let $k=2\alpha+1$.
Our algorithm is composed of an initialization phase and three different phases of girth approximations. Roughly speaking, Phase 1 handles the first ``half level'' using ideas from \cite{Ducoffe19}, and Phases 2,3 use ideas from \cite{kadria2023improved} for the rest of the levels.

In the initialization phase, we initialize all the data structures needed for the algorithm; this phase is similar to $\Initialize$  from \cite{kadria2023improved}. For completeness, we describe $\Initialize$ in \Cref{AP-In}, with one major difference: we create a hierarchy of vertex sets using non-uniform sampling. We have $A_0=V$, $A_{\alpha+1}=\emptyset$, $A_1 = \Sample(A_0,n^{-1/(2\alpha+1)})$, and \\$A_i = \Sample(A_{i-1},n^{-2/(2\alpha+1)})$, for every $2\le i \le \alpha$. The reason for these different probabilities is to create a first ``half level''. Using this ``half level'', the algorithm can address odd values of $k$. In this phase, the current girth estimation $\hat{g}$ is initialized to $\infty$, the current shortest cycle $\hat{C}$ is set to $\emptyset$, and both the distance hash table $d$ and the predecessor hash table $\pi$ (used to reconstruct shortest paths) are initialized as empty hash tables. In addition, the algorithm also computes $\delta(p_i(u),u)$ for every $0\le i \le \alpha$, and calls $\Preprocess(G)$, as required by \cite{kadria2023improved}, to efficiently implement $\ClusterOrCycle$.

In Phase 1, for every vertex $u \in V$, the algorithm computes the subgraph $G[B_0(u)]$ using \Cref{L-Create_B_S} and calls $\MinCycle(u,G[B_0(u)])$. If the length of a shortest cycle that goes through $u$ in $G[B_0(u)]$ is smaller than  $\hat{g}$, then the algorithm updates $\hat{C}$ and $\hat{g}$.

In Phase 2, for every vertex $u \in A_1$, the algorithm computes $\ClusterOrCycle(u)$.
If a cycle shorter than the current estimate $\hat{g}$  is detected by  $\ClusterOrCycle(u)$ in $\Cl(u)$, then the algorithm updates $\hat{C}$  and $\hat{g}$.

Finally, in phase 3, for every edge $(v, w) \in E$ and every $0 \le i \le \alpha$, the algorithm checks whether the edge $(v, w)$ closes a cycle in the shortest paths tree rooted at $u = p_i(v)$. If $(v,w)$ indeed closes a cycle and the length of this cycle is smaller than $\hat{g}$, then the algorithm updates $\hat{C}$ and $\hat{g}$.

The algorithm checks whether $(v, w)$ actually closes a cycle rather than merely retracing a path in the shortest paths tree rooted at $u$  by  
verifying that $\pi(u, v) \ne (w, v)$ and $\pi(u, w) \ne (v, w)$. 
If $\pi(u, v) \ne (w, v)$ and $\pi(u, w) \ne (v, w)$, a cycle is indeed formed, and the length of the cycle is bounded from above by $\hat{g}'=d(u,v)+\ell(v,w)+\delta(u,w)$. If $\hat{g}'<\hat{g}$, the algorithm updates $\hat{g}$ accordingly. In this case, we succinctly represent the discovered cycle $\hat{C}$ with the triplet $(u,v,w)$. 

At the end of the algorithm, to obtain a cycle from the triplet $(u,v,w)$, the algorithm computes $u'$, the Lowest Common Ancestor (LCA) of~$v$ and~$w$ in the shortest paths tree rooted at~$u$, and returns the simple cycle $P(u',v) \cup P(u',w) \cup (v,w)$.
The algorithm returns $\hat{g}$ as the estimated girth, along with $\hat{C}$ as the corresponding cycle.

% \begin{algorithm2e}[t]
% \caption{$\Cycle(G,\alpha)$}\label{A-Cycle}
% \tcc{Initialization phase}
% $\hat{g}\gets \infty$, $\hat{C}\gets \emptyset$, $d\gets\HashTable()$, $\pi\gets\HashTable()$\\
% $A_0\gets V$ ; $A_{\alpha+1}=\emptyset$ ; $A_1 \gets \Sample(A_0,n^{-1/(2\alpha+1)})$ \\
% \lFor{$i \in [2,\alpha]$}{
%     $A_i\gets \Sample(A_{i-1},n^{-2/(2\alpha+1)})$
% }
% $\Preprocess(G)$\\
% Compute $p_i(u)$ and $d(u,p_i(u))$ for every $\pair{u,i}\in \pair{V,[\alpha]}$\\

% \tcc{Phase 1:}
% \For{$u\in V$}{
%     Compute $G[B_0(u)]$ using~\Cref{L-Create_B_S}\\
%     $\pair{\hat{g}',\hat{C}'}\gets \MinCycle(u,G[B_0(u)])$\\
%     \If{$\hat{g}'\le \hat{g}$}{
%         $\hat{g}\gets \hat{g}'$ ; $\hat{C}\gets \hat{C}'$
%     }
% }

% \tcc{Phase 2:}
% \For{$u\in A_1$}{
%     $\pair{\hat{g}',\hat{C}'}\gets \ClusterOrCycle(u)$\\
%     \If{$\hat{g}'\le \hat{g}$}{
%         $\hat{g}\gets \hat{g}'$ ; $\hat{C}\gets \hat{C}'$
%     }
% }

% \tcc{Phase 3:}
% \For {$(v,w)\in E$}{ \label{L-Alg-Loop2}
%     \For{$i\gets 0$ {\bf to} $\alpha$}
%         {
%         $u\gets p_i(v)$ ;
%         $\hat{g}' \gets d(u,v) + \ell(v,w) + d(u,w)$ \;
%         \If {$\hat{g}' < \hat{g}$ {\bf and} $\pi(u,v) \neq (w,v)$ {\bf and} $\pi(u,w) \neq (v,w)$ }{ 
%         \label{L-Cycle-If-pi}
%             $\hat{g} \gets \hat{g}'$ ;
%             $\hat{C}\gets (u,v,w)$ \;
%         }
%     }
% }
% \lIf {$\hat{C}$ is a triplet $(u,v,w)$} {$u'\gets LCA(u,v,w)$; $\hat{C}=P(u',v) \cup P(u',w) \cup (v,w)$}
% \Return $\pair{\hat{g},\hat{C}}$
% \end{algorithm2e}

Next, we bound the running time of the algorithm and show that $\Cycle(G,k)$ takes $\Ot(m + n^{1 + 2/k}) = \Ot(m + n^{1 + 2/(2\alpha + 1)})$ time, since $k = 2\alpha + 1$.
\begin{lemma}[\Cref{L-alg-runtime} in \Cref{S-algo-full}]
    $\Cycle(G,k)$ takes $O((m+kn^{1+2/(2\alpha+1)})\log{n} + \alpha m)$ time.
\end{lemma}
\begin{proof}
    See \Cref{L-alg-runtime} for the complete proof.
\end{proof}

Next, we prove the following lemma, which is a main tool used in our correctness proof. 

\begin{lemma}[\Cref{L-bound-h2} in \Cref{S-algo-full}]
Let $C$ be a shortest cycle; that is, $\ell(C)=g$. 
Let $u \in C$, and let $(x, y)$ be the farthest edge from $u$ in $C$, i.e., $(x,y)=\arg\max_{(x,y)\in C}(\delta(u,(x,y)))$.

Then either $\hat{g} \le 2g$ or $\min(\delta(A_2,x),\delta(A_2,y)) \le g$.
\end{lemma}
\begin{proof}
    Let $u \in C$, and let $(x, y)$ be the farthest edge from $u$ in $C$. Without loss of generality, assume that $\delta(u, x) \ge \delta(u, y)$.  
    If $C\subseteq G[B_0(u)]$, then we have \\$\hat{g} \le \ell(\MinCycle(u,G[B_0(u)])) = g \le 2g$, and the claim holds. Therefore, we assume that $C \not\subseteq G[B_0(u)]$.
    
    Since $C \not\subseteq G[B_0(u)]$, there exists a vertex $w \in C$ such that $w \notin B_0(u)$. 
    From the definition of $B_0(u)$, we have $\delta(u, A_1) \le \delta(u, w)$. Since $x$ is the farthest vertex from $u$ in $C$, we get:
    \begin{equation}\label{e-0-1}
        \delta(u,A_1) \le \delta(u,w) \le \delta(u,x)
    \end{equation}
    
    Since $C$ is a shortest cycle and $x\in C$, we have that $C=P(u,x)\cup P_2(u,x)$, and therefore $\ell(C)=\ell(P(u,x))+\ell(P_2(u,x))$. Thus,
    \begin{equation}\label{e-1-1}
        g=\ell(C)=\ell(P(u,x))+\ell(P_2(u,x))=\delta(u,x)+\ell(P_2(u,x)).
    \end{equation}

    Let $r = \delta(u, A_1) + \ell(P_2(u, x))$. We show that:
    \begin{cclaim}[\Cref{C-C-in-G_r} in \Cref{S-algo-full}]
        $C \subseteq G_r(p_1(u))$
    \end{cclaim}
    \begin{proof}
        See \Cref{C-C-in-G_r} in \Cref{S-algo-full} for the complete proof.
    \end{proof}
        % In order to prove that $C \subseteq G_r(p_1(u))$ we prove that for every $e\in C$ we have that $\delta(p_1(u),e)\le r$ and therefore $e\in G_r(p_1(u))$, as required.
        
        % Since $C=P(u,x)\cup P_2(u,x)$ we have that either   $e\in P(u,x)$ or   $e\in P_2(u,x)$. If $e\in P(u,x)$ then we have that $\delta(u,e)\le \ell(P(u,x)) \le \ell(P_2(u,x))$, and therefore $\delta(p_1(u),e) \stactri \le \delta(A_1,u)+\ell(P_2(u,x))=r$, as required.\footnote{Throughout the paper, $\stactri\le $ is a step that follows from the triangle inequality.}
        % If $e\in P_2(u,x)$ then we have that $\delta(u,e) \le \ell(P_2(u,x))$, and therefore $\delta(p_1(u),e) \stactri \le \delta(A_1,u) + \ell(P_2(u,x))$, as required.
    % \end{proof}

    We divide the rest of the proof into the case that $C \subseteq \Cl(p_1(u))$, and the case that $C \not\subseteq \Cl(p_1(u))$.
    If $C \subseteq \Cl(p_1(u))$, then together with \Cref{C-C-in-G_r} we have that $C \subseteq G_r(p_1(u)) \cap \Cl(p_1(u))$, and therefore from \Cref{L-ClusterOrCycle} we have that:
    \begin{align*}
        \hat{g} &\le \ell(\ClusterOrCycle(p_1(u))) \le 2r=2(\delta(p_1(u),u) + \ell(P_2(u,x))) \\
        &\stackrel{\ref{e-0-1}}{\le} 2(\delta(u,x)+\ell(P_2(u,x))) \stackrel{\ref{e-1-1}}{\le} 2g,
    \end{align*}
    as required.

    If $C \not\subseteq \Cl(p_1(u))$, then we must have a cycle edge $(s,t)\in C$ such that $(s,t)\notin \Cl(p_1(u))$.
    If $(s,t)\neq (x,y)$ (recall that $(x, y)$ is the farthest edge from $u$ in $C$), then since $C$ is a shortest cycle we have that $(s,t)\subseteq P(u,t)$ or $(t,s)\subseteq P(u,s)$, and therefore $\delta(u,(s,t))=\delta(u,s)$ or $\delta(u,(s,t))=\delta(u,t)$. Wlog, we assume that $\delta(u,(s,t))=\delta(u,s)+\ell(s,t)=\delta(u,t)$.
    Since $(s,t)\notin \Cl(p_1(u))$, it follows from the definition of $\Cl(p_1(u))$ that
    \begin{equation}\label{e-A-2-t-1}
        \delta(A_2,t) \le \delta(p_1(u),s)+\ell(s,t) \stactri\le \delta(A_1,u)+\delta(u,s) + \ell(s,t) = \delta(A_1,u) + \delta(u,t)
    \end{equation}
    
    Since $(s,t)\neq (x,y)$ it follows that $(s,t) \in P(u,x) \cup P(u,y)$. (See \Cref{fig:c_6_1} for an illustration of this case.)
    If $(s, t) \in P(u, y)$, then $\delta(u, y) = \delta(u, t) + \delta(t, y)$. Therefore:
    \begin{align*}
        \delta(A_2, y) &\stactri\le \delta(y,t)+\delta(A_2,t) \stackrel{\ref{e-A-2-t-1}}\le \delta(y,t)+\delta(p_1(u),u) + \delta(u,t)
        \\&\stackrel{x\notin B_0(u)}\le \delta(u,t)+\delta(t,y)+\delta(u,x) \stackrel{t\in P(u,y)}\le \delta(u,y)+\delta(u,x) \le g/2+g/2=g,
    \end{align*}
    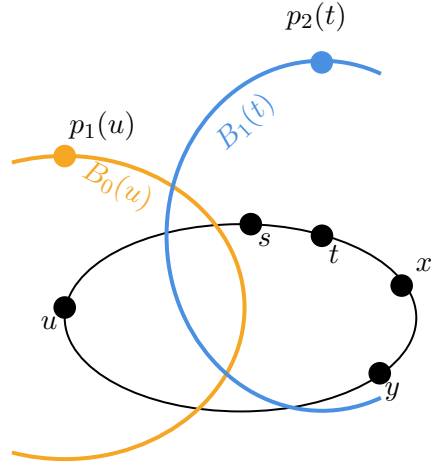
\begin{figure}
    \centering
    \tikzset{every picture/.style={line width=0.75pt}} %set default line width to 0.75pt        
    \tikzset{every picture/.style={line width=0.75pt}} %set default line width to 0.75pt        

\begin{tikzpicture}[x=0.75pt,y=0.75pt,yscale=-1,xscale=1]
%uncomment if require: \path (0,261); %set diagram left start at 0, and has height of 261

%Shape: Ellipse [id:dp2827032205151191] 
\draw   (247.16,170) .. controls (247.16,144.04) and (286.66,123) .. (335.38,123) .. controls (384.11,123) and (423.61,144.04) .. (423.61,170) .. controls (423.61,195.96) and (384.11,217) .. (335.38,217) .. controls (286.66,217) and (247.16,195.96) .. (247.16,170) -- cycle ;
%Shape: Arc [id:dp08222238404625459] 
\draw  [draw opacity=0][line width=1.5]  (220.68,237.59) .. controls (229.05,239.75) and (237.95,240.92) .. (247.16,240.92) .. controls (297.01,240.92) and (337.41,206.85) .. (337.41,164.83) .. controls (337.41,122.81) and (297.01,88.75) .. (247.16,88.75) .. controls (237.95,88.75) and (229.05,89.91) .. (220.68,92.08) -- (247.16,164.83) -- cycle ; \draw  [color={rgb, 255:red, 245; green, 166; blue, 35 }  ,draw opacity=1 ][line width=1.5]  (220.68,237.59) .. controls (229.05,239.75) and (237.95,240.92) .. (247.16,240.92) .. controls (297.01,240.92) and (337.41,206.85) .. (337.41,164.83) .. controls (337.41,122.81) and (297.01,88.75) .. (247.16,88.75) .. controls (237.95,88.75) and (229.05,89.91) .. (220.68,92.08) ;  
%Shape: Arc [id:dp9764287245762697] 
\draw  [draw opacity=0][line width=1.5]  (405.86,210.13) .. controls (396.74,214.34) and (386.75,216.67) .. (376.27,216.67) .. controls (333.19,216.67) and (298.27,177.34) .. (298.27,128.83) .. controls (298.27,80.32) and (333.19,41) .. (376.27,41) .. controls (386.75,41) and (396.74,43.32) .. (405.86,47.54) -- (376.27,128.83) -- cycle ; \draw  [color={rgb, 255:red, 74; green, 144; blue, 226 }  ,draw opacity=1 ][line width=1.5]  (405.86,210.13) .. controls (396.74,214.34) and (386.75,216.67) .. (376.27,216.67) .. controls (333.19,216.67) and (298.27,177.34) .. (298.27,128.83) .. controls (298.27,80.32) and (333.19,41) .. (376.27,41) .. controls (386.75,41) and (396.74,43.32) .. (405.86,47.54) ;  
%Shape: Circle [id:dp74616805340857] 
\draw  [fill={rgb, 255:red, 0; green, 0; blue, 0 }  ,fill opacity=1 ] (241.99,164.83) .. controls (241.99,161.98) and (244.31,159.67) .. (247.16,159.67) .. controls (250.01,159.67) and (252.33,161.98) .. (252.33,164.83) .. controls (252.33,167.69) and (250.01,170) .. (247.16,170) .. controls (244.31,170) and (241.99,167.69) .. (241.99,164.83) -- cycle ;
%Shape: Circle [id:dp7874153787103315] 
\draw  [fill={rgb, 255:red, 0; green, 0; blue, 0 }  ,fill opacity=1 ] (335.38,123) .. controls (335.38,120.15) and (337.7,117.83) .. (340.55,117.83) .. controls (343.4,117.83) and (345.72,120.15) .. (345.72,123) .. controls (345.72,125.85) and (343.4,128.17) .. (340.55,128.17) .. controls (337.7,128.17) and (335.38,125.85) .. (335.38,123) -- cycle ;
%Shape: Circle [id:dp26385009266109005] 
\draw  [fill={rgb, 255:red, 0; green, 0; blue, 0 }  ,fill opacity=1 ] (371.11,128.83) .. controls (371.11,125.98) and (373.42,123.67) .. (376.27,123.67) .. controls (379.13,123.67) and (381.44,125.98) .. (381.44,128.83) .. controls (381.44,131.69) and (379.13,134) .. (376.27,134) .. controls (373.42,134) and (371.11,131.69) .. (371.11,128.83) -- cycle ;
%Shape: Circle [id:dp40079864585505265] 
\draw  [fill={rgb, 255:red, 0; green, 0; blue, 0 }  ,fill opacity=1 ] (411.11,153.83) .. controls (411.11,150.98) and (413.42,148.67) .. (416.27,148.67) .. controls (419.13,148.67) and (421.44,150.98) .. (421.44,153.83) .. controls (421.44,156.69) and (419.13,159) .. (416.27,159) .. controls (413.42,159) and (411.11,156.69) .. (411.11,153.83) -- cycle ;
%Shape: Circle [id:dp5936497669172328] 
\draw  [fill={rgb, 255:red, 0; green, 0; blue, 0 }  ,fill opacity=1 ] (400.11,197.83) .. controls (400.11,194.98) and (402.42,192.67) .. (405.27,192.67) .. controls (408.13,192.67) and (410.44,194.98) .. (410.44,197.83) .. controls (410.44,200.69) and (408.13,203) .. (405.27,203) .. controls (402.42,203) and (400.11,200.69) .. (400.11,197.83) -- cycle ;
%Shape: Circle [id:dp960849577468932] 
\draw  [color={rgb, 255:red, 245; green, 166; blue, 35 }  ,draw opacity=1 ][fill={rgb, 255:red, 245; green, 166; blue, 35 }  ,fill opacity=1 ] (241.99,88.83) .. controls (241.99,85.98) and (244.31,83.67) .. (247.16,83.67) .. controls (250.01,83.67) and (252.33,85.98) .. (252.33,88.83) .. controls (252.33,91.69) and (250.01,94) .. (247.16,94) .. controls (244.31,94) and (241.99,91.69) .. (241.99,88.83) -- cycle ;
%Shape: Circle [id:dp9217005175303342] 
\draw  [color={rgb, 255:red, 74; green, 144; blue, 226 }  ,draw opacity=1 ][fill={rgb, 255:red, 74; green, 144; blue, 226 }  ,fill opacity=1 ] (370.99,41.83) .. controls (370.99,38.98) and (373.31,36.67) .. (376.16,36.67) .. controls (379.01,36.67) and (381.33,38.98) .. (381.33,41.83) .. controls (381.33,44.69) and (379.01,47) .. (376.16,47) .. controls (373.31,47) and (370.99,44.69) .. (370.99,41.83) -- cycle ;

% Text Node
\draw (233.74,168) node [anchor=north west][inner sep=0.75pt]   [align=left] {$\displaystyle u$};
% Text Node
\draw (421.83,138.83) node [anchor=north west][inner sep=0.75pt]   [align=left] {$\displaystyle x$};
% Text Node
\draw (406.27,200.83) node [anchor=north west][inner sep=0.75pt]   [align=left] {$\displaystyle y$};
% Text Node
\draw (248,65) node [anchor=north west][inner sep=0.75pt]   [align=left] {$\displaystyle p_{1}( u)$};
% Text Node
\draw (357,10) node [anchor=north west][inner sep=0.75pt]   [align=left] {$\displaystyle p_{2}( t)$};
% Text Node
\draw (342.55,126) node [anchor=north west][inner sep=0.75pt]   [align=left] {$\displaystyle s$};
% Text Node
\draw (378.27,131.83) node [anchor=north west][inner sep=0.75pt]   [align=left] {$\displaystyle t$};
% Text Node
\draw (318.48,78.21) node [anchor=north west][inner sep=0.75pt]  [color={rgb, 255:red, 74; green, 144; blue, 226 }  ,opacity=1 ,rotate=-316.81] [align=left] {$\displaystyle B_{1}( t)$};
% Text Node
\draw (259.73,86.88) node [anchor=north west][inner sep=0.75pt]  [color={rgb, 255:red, 245; green, 166; blue, 35 }  ,opacity=1 ,rotate=-26.93] [align=left] {$\displaystyle B_{0}( u)$};

\end{tikzpicture}

    \caption{$\ell(C)=g$, $u\in C$, $(x,y)\in C$ farthest edge from $u$, $(s,t)\in C$ such that $(s,t)\notin \Cl(p_1(u))$ and $(s,t)\neq (x,y)$.}
    \label{fig:c_6_1}
\end{figure}

    where the last inequality follows from the fact that $\delta(u, z) \le \frac{g}{2}$ for every $z \in C$, as required. 

    If $(s, t) = (y, x)$, then $(y, x) \notin \Cl(p_1(u))$, and by the definition of $\Cl(p_1(u))$ we have $\delta(A_2,x) \le \delta(p_1(u),y)+\ell(y,x)$. Therefore:
    \begin{align*}
        \delta(A_2,x) &\le \delta(p_1(u),y) + \ell(y,x) \stactri\le \delta(p_1(u),u)+\delta(u,y)+\ell(y,x) \\
        &\stackrel{\ref{e-0-1}}\le \delta(u,x)+\ell(P_2(u,x)) \stackrel{\ref{e-1-1}}{=} g,
    \end{align*}
    as required.
\end{proof}

Next, using \Cref{L-bound-h2} we bound the value of $\hat{g}$ from above by $\frac{2}{3}(2\alpha+1)g$.
Let $C$ be a shortest cycle in $G$, i.e. $\ell(C)=g$. Recall that $M(C)$ is the length of a longest edge in $C$. 
To prove that $\hat{g} \le \frac{2}{3}(2\alpha+1)g$, we consider the following two cases: the case that $M(C) \le g/3$ (\Cref{L-Corr-small}) and the case that $M(C) > g/3$  (\Cref{L-corr-big}).\footnote{We remark that the distinction between $M(C) \le g/3$ and $M(C) > g/3$ was also considered in the correction proof of \cite{kadria2023improved}.}
If $M(C) \le g/3$, then we show that in the first two loops of the algorithm, a cycle of length at most $\frac{2}{3}(2\alpha+1)g$ is found. Otherwise, if $M(C) > g/3$, we consider a longest edge $(u, u') \in C$ and show that either a short cycle is found during the first two loops, or that while considering the edge $(u,u')$ in the third loop of the algorithm a cycle of length at most $\frac{2}{3}(2\alpha+1)g$ is found. For the complete proofs, see \Cref{S-algo-full}. \Cref{T-main-odd} follows from \Cref{L-Corr-small,L-corr-big,L-alg-runtime}.
% We start with the case where $M(C) \le g/3$. 
% \begin{lemma}\label{L-Corr-small}
%     If $M(C) \le g/3$ then $\hat{g}\le \frac{2}{3}(2\alpha+1)g$
% \end{lemma}
% \begin{proof}
% Let $h_i(C)=\min_{w\in A_i}(\delta(w,C))$. (See \Cref{fig:h_i} for an illustration of the $h_i$.)
\begin{figure}[htbp]
    \centering
    \begin{minipage}{\linewidth}
        \centering
        \tikzset{every picture/.style={line width=0.75pt}} % set default line width
        % This forces the input figure to scale down to the minipage width
        \resizebox{\linewidth}{!}{%
            \input{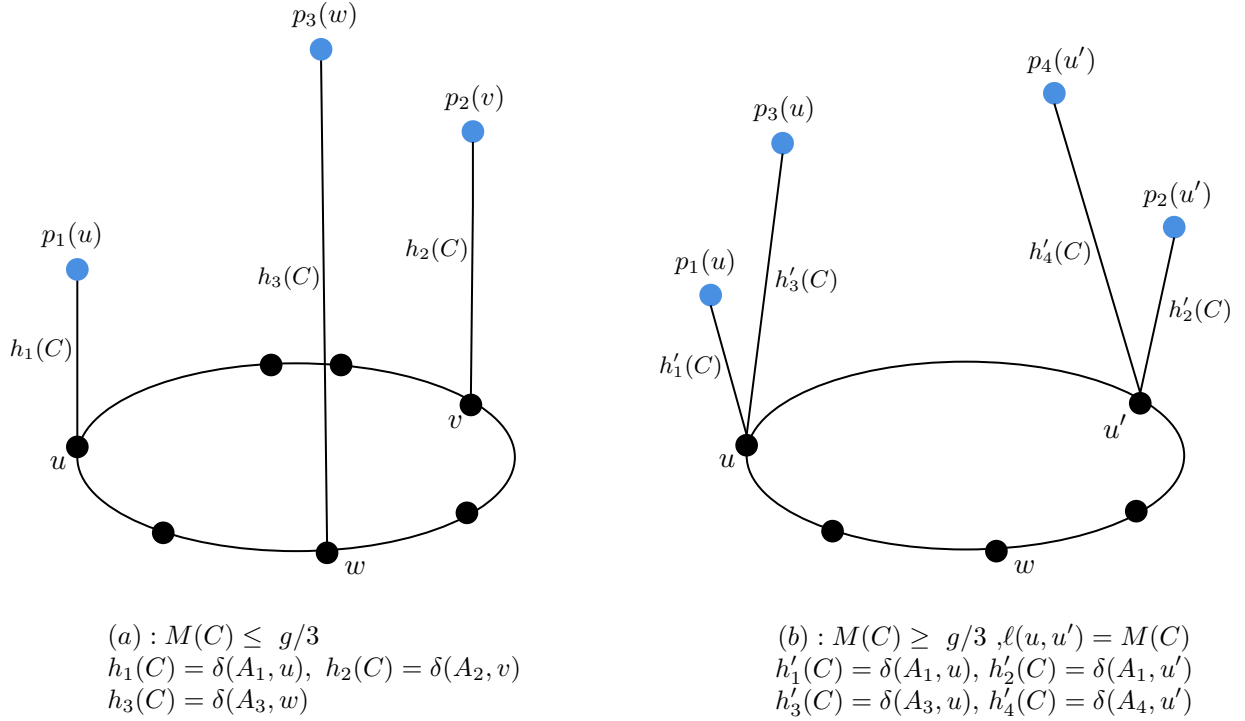}%
        }
        \caption{Illustration of $h_i(C)$ and $h_i'(C)$, where $C$ is a shortest cycle in $G$.}
        \label{fig:h_i}
    \end{minipage}
\end{figure}

\section{Conditional lower bound for $(5/3-\epsilon)$-girth approximation under the \StrongerThreeSum hypothesis}

In this section, we show a girth approximation lower bound under the \StrongerThreeSum hypothesis, which is defined as follows.
\begin{hypothesis}[\StrongerThreeSum]
    There is no $O(n^{2-\eps})$ time algorithm for \ThreeSum on a set $A$ where $A\subseteq [\pm O(n^2)]$ and $A$ is a Sidon set.
\end{hypothesis}

We remark that although this hypothesis seems much stronger than the Strong \ThreeSum hypothesis, it is still plausible. In fact, \cite{jin2023removing} explicitly asks (Open problem 5) whether their techniques for fine-grained reductions for \ThreeSum in Sidon Sets can be used to prove the \StrongerThreeSum hypothesis assuming the Strong \ThreeSum hypothesis.

\cite{jin2023removing} proved a conditional lower bound for detecting a triangle under the Strong \ThreeSum hypothesis. 
Next, we prove a conditional lower bound for detecting $C_{\le 4}$, i.e. whether a cycle contains a triangle or a $C_4$. 

The reduction first follows the same steps as the reduction of Theorem 1.14 of \cite{jin2023removing}. 
Then, we show how to take a small number of subgraphs of the graphs created in the reduction of \cite{jin2023removing}
so that if the input set $A$ does not have a non-trivial solution to $a+b=c+d$ then the subgraphs created have no $4$-cycles. Moreover, as before, any triangle corresponds to a $3$SUM solution, and if $A$ has a $3$SUM solution, one of the subgraphs is guaranteed to have a triangle.
Therefore, detecting a triangle in a sparse $4$-cycle free graph is hard.
We prove:

% \avi{Theorem 2, p15: something is missing in the implementation details: if I scan all the y's between 0 and pqr, then the runtime should depend on pqr, not just on r.}
\Reminder{T-Stronger-3sum-LB}
\begin{proof}
    Assume for the sake of contradiction that for some $\eps>0$ there exists an $O(n^{1.5-\eps})$ time algorithm $Alg$ for $C_{\le 4}$ detection in graphs where $\text{deg}(u)=O(n^{1/4})$, for every $u\in V$. We will show that there exists an $O(n^{2-\eps})$ time algorithm for \ThreeSum on a (Sidon) set $A$ where $|A|=n$, $A\subseteq [\pm C_1 n^2]$, for some constant $C_1$ and there are no non-trivial solutions to $a+b=c+d$ with $a,b,c,d\in A$.

    First, we describe the construction of \cite{jin2023removing}.
    Let $A\subseteq [\pm n^2]$ be the input for \ThreeSum. Let $p,q,r \in [2C_1n^{2/3}\log^2{n},4C_1n^{2/3}\log^2{n}]$ be three distinct primes, sampled uniformly at random. Let $A'=\{a\pmod{pqr} \mid a\in A\}$. Let $U$ (resp. $V$,$W$) be a vertex set identified by all numbers $[0,pqr)$ that are congruent to $0$ mod $p$ (resp. $q$,$r$). We add an edge between $u\in U$ and $v\in V$ if $v-u\pmod{pqr}\in A'$, between $v\in V$ and $w\in W$ if $w-v\pmod{pqr}\in A'$ and between $w\in W$ and $u\in U$ if $u-w\pmod {pqr}\in A'$. See \Cref{fig:h_G_uvw} for an illustration of this graph. \begin{figure}
    \centering
    \tikzset{every picture/.style={line width=0.75pt}} %set default line width to 0.75pt        
    \tikzset{every picture/.style={line width=0.75pt}} %set default line width to 0.75pt        

\begin{tikzpicture}[x=0.75pt,y=0.75pt,yscale=-1,xscale=1]

%Shape: Circle [id:dp4300239290897244] 
\draw   (188.33,203.83) .. controls (188.33,190.02) and (199.53,178.83) .. (213.33,178.83) .. controls (227.14,178.83) and (238.33,190.02) .. (238.33,203.83) .. controls (238.33,217.63) and (227.14,228.83) .. (213.33,228.83) .. controls (199.53,228.83) and (188.33,217.63) .. (188.33,203.83) -- cycle ;
%Shape: Circle [id:dp7727679704550483] 
\draw   (290.33,97.83) .. controls (290.33,84.02) and (301.53,72.83) .. (315.33,72.83) .. controls (329.14,72.83) and (340.33,84.02) .. (340.33,97.83) .. controls (340.33,111.63) and (329.14,122.83) .. (315.33,122.83) .. controls (301.53,122.83) and (290.33,111.63) .. (290.33,97.83) -- cycle ;
%Shape: Circle [id:dp736056338843742] 
\draw   (404.33,201.83) .. controls (404.33,188.02) and (415.53,176.83) .. (429.33,176.83) .. controls (443.14,176.83) and (454.33,188.02) .. (454.33,201.83) .. controls (454.33,215.63) and (443.14,226.83) .. (429.33,226.83) .. controls (415.53,226.83) and (404.33,215.63) .. (404.33,201.83) -- cycle ;
%Straight Lines [id:da567597020861471] 
\draw    (226.33,181.83) -- (295.33,113.83) ;
%Straight Lines [id:da10139563159690912] 
\draw    (238.33,203.83) -- (404.33,201.83) ;
%Straight Lines [id:da5562878878913574] 
\draw    (336.33,111.83) -- (414.33,180.83) ;
%Rounded Rect [id:dp4435870219796456] 
\draw   (17,41.2) .. controls (17,33.91) and (22.91,28) .. (30.2,28) -- (225.13,28) .. controls (232.42,28) and (238.33,33.91) .. (238.33,41.2) -- (238.33,80.8) .. controls (238.33,88.09) and (232.42,94) .. (225.13,94) -- (30.2,94) .. controls (22.91,94) and (17,88.09) .. (17,80.8) -- cycle ;

% Text Node
\draw (203.33,191.83) node [anchor=north west][inner sep=0.75pt]  [font=\LARGE] [align=left] {$\displaystyle U$};
% Text Node
\draw (304.33,85.83) node [anchor=north west][inner sep=0.75pt]  [font=\LARGE] [align=left] {$\displaystyle V$};
% Text Node
\draw (418.33,189.83) node [anchor=north west][inner sep=0.75pt]  [font=\LARGE] [align=left] {$\displaystyle W$};
% Text Node
\draw (213.21,160.21) node [anchor=north west][inner sep=0.75pt]  [rotate=-317] [align=left] {$\displaystyle v-u\in A'$};
% Text Node
\draw (269.3,207.99) node [anchor=north west][inner sep=0.75pt]  [rotate=-359.77] [align=left] {$\displaystyle u-w\in A'$};
% Text Node
\draw (360.45,105.02) node [anchor=north west][inner sep=0.75pt]  [rotate=-41.27] [align=left] {$\displaystyle w-v\in A'$};
% Text Node
\draw (27,35) node [anchor=north west][inner sep=0.75pt]  [font=\footnotesize] [align=left] {$\displaystyle U\leftarrow \{i=0\bmod p\ \mid i\in [ 0,pqr)\}$\\$\displaystyle V\leftarrow \{i=0\bmod q\ \mid i\in [ 0,pqr)\}$\\$\displaystyle W\leftarrow \{i=0\bmod r\ \mid i\in [ 0,pqr)\}$};

\end{tikzpicture}

    \caption{Illustration of \Cref{T-Stronger-3sum-LB}.}
    \label{fig:h_G_uvw}
\end{figure}

    \cite{jin2023removing} shows that these edges can be added efficiently: for each $a \in A'$,
    if $v - u = a \pmod{pqr}$, then $v - u = a \pmod q$. Also, by construction, $u = 0 \pmod{p}$ and $v = 0 \pmod{q}$, which imply $u = -a \pmod{q}$. For every fixed $a\in A'$ and $y\in [0,\ldots, pqr-1)$, if $u'=0\pmod p, u'=-a\pmod q, u'=y\pmod r,$ then $u=(u'\pmod {pqr})$ is uniquely determined by the Chinese Remainder Theorem. Hence, we can go over all $a\in A', y\in [0,\ldots, pqr-1)$ which determine $u\in U$ and then add an edge to $(u+a\pmod {pqr})\in V$.
    It thus takes  $\Ot(n \cdot r) = \Ot(n^{5/3})$ time to add all the edges. (The addition of the edges in $V\times W$ and $W\times U$ is done similarly.)
    Finally, we remove all the edges in the graph whose degree is larger than $Cn^{1/3}$.

    \cite{jin2023removing} shows that the graph has a triangle if and only if there is a solution for \ThreeSum in $A$; for completeness, we prove this for our construction in the appendix.
    \begin{cclaim} \label{C-Correctness-jin2023}
        The graph has a triangle if and only if there is a solution for \ThreeSum in $A$.
    \end{cclaim}
    \begin{proof}
        See \Cref{S-LB-full} for the complete proof.  
    \end{proof}

    If $A$ had a  solution $x+y+z=0$ where $x=z$, then $(A, 2A)$ would have a $2$SUM solution. Since $2$SUM can be solved in $O(n\log n)$ time by sorting, we can solve \ThreeSum in $A$ in that much time. Hence, we assume that for any \ThreeSum solution in $A$, we have $x\neq z$. 
    Since this means that $(x-z)\neq 0$, and since $(x-z)$ has $O(\log n)$ prime factors, we can assume that whp our random prime $p$ does not divide $x-z$ and hence $x\neq z\pmod p$. 
    
    Now, notice that if $A$ has a \ThreeSum solution $x+y+z=0$, then whp in one of our trials, there is a triangle $(u,v,w)$ in the graph where \[u=0\pmod p, v=x\pmod p, w=-z\pmod p.\] 
    
    Since we can assume that $x\neq z \pmod p$ in any \ThreeSum solution, we can reduce \ThreeSum to $O(\log n)$ instances of subgraphs of our graphs above, where for every $v\in V, w\in W$ we have that $v\neq -w \pmod p$ and such that we are guaranteed that whp one of the subgraphs will contain a triangle if $A$ had a \ThreeSum.

    We do this as follows. Fix one of our graphs $G$ from the trials above. For each $i\in [0,\ldots,\log p)$ and each choice of $(b,b')\in \{(0,1),(1,0)\}$, we create a subgraph $G_{bb'}$ of $G$ by only keeping those $v\in V$ such that the $i$th bit of $v\pmod p$ is $b$, and those $w\in W$ such that the $i$th bit of $-w\pmod p$ is $b'$.
    Since a triangle $(u,v,w)$ corresponding to a \ThreeSum solution has $v\neq -w \pmod p$, there is some bit $i$ for which the $i$th bits of $v\pmod p$ and $-w\pmod p$ ($b$ and $b'$) differ. For the corresponding subgraph $G_{bb'}$, both $v$ and $w$ appear and hence $(u,v,w)$ is a triangle. Meanwhile, for every graph $G_{bb'}$, for every $v,w$ we have that $v\neq -w\pmod p$.
    
    Now let's focus on any one of the graphs $G_{bb'}$ mentioned above.
    Suppose, for contradiction, that $G_{bb'}$ contains a $4$-cycle. There are two types of $4$-cycles. The first is w.l.o.g of the form $u,u'\in U, v,v'\in V$ that goes $u\rightarrow v\rightarrow u'\rightarrow v'\rightarrow u$ (or $(U,V)$ can be $(V,W)$ or $(W,U)$). The second is wlog of the form $u\rightarrow v\rightarrow u'\rightarrow w\rightarrow u$, for $u,u'\in U$, $v\in V$, $w\in W$ (or some permutation of $U,V,W$). 

    Let's consider the first type of $4$-cycle: $u\rightarrow v\rightarrow u'\rightarrow v'\rightarrow u$. By construction, we have that $v-u=a\pmod {pqr}, v'-u=a'\mod (pqr), v-u'=a''\mod (pqr), v'-u'=a'''\pmod {pqr}$. As $v\neq v', u\neq u'$ we get that  $a\neq a'$, $a''\neq a'''$, $a\neq a''$ and $a'\neq a'''$. 

    However, we have that $a-a''=v-v'=a'-a'''\mod (pqr)$, so that $a+a'''=a'+a''\mod (pqr)$. Since $pqr$ is chosen to be large enough, we also have that $a+a'''=a'+a''$ for the corresponding integers in $A$. We showed that $a\neq a',a''$ and $a'''\neq a',a''$ so we have a non-trivial solution to $x+y=z+w$, a contradiction to $A$ being a Sidon set.

    For the second type of $4$-cycle $u-v-u'-w-u$ we obtain a similar contradiction but with a bit more work.
    We get $u-w=a\mod (pqr), v-u=a'\mod (pqr), u'-w=a''\mod (pqr), v-u'=a'''\mod (pqr)$.
    We have that $a+a'=v-w=a''+a'''\mod (pqr)$ so there is a solution in $A$ $$a+a'=a''+a'''.$$

    Again, since $u\neq u'$ we get that $a\neq a''$, $a'\neq a'''$.

    Suppose for the sake of contradiction that $a=a'''$ (a similar argument can be made for $a'=a''$). Then from the definition of $a$ and $a'''$ we get that $(u-w)=(v-u')\mod (pqr)$ and hence $$u+u'=v+w\mod (pqr).$$
    
    However, consider this equation modulo $p$. Since for all $u\in U$ we have $u=0\pmod p$, we obtain $u+u'=0\pmod p$. Meanwhile, $v+w\neq 0\pmod p$ because in $G_{bb'}$ we ensured that for all $v,w$, $v\neq -w\pmod p$, a contradiction. 

    Thus, we must have that $a\neq a'''$ and $a'\neq a''$, in addition to the fact that $a\neq a''$ and $a'\neq a'''$, and therefore $a+a'=a''+a'''$ is a non-trivial solution to $x+y=z+w$, a contradiction to $A$ being a Sidon set. Thus, such cycles cannot exist in $G_{bb'}$.

    This means that in $G_{bb'}$ we have removed the $4$-cycles with two nodes in $U$.
    There could be $4$-cycles with two nodes in $V$ (and one in $U,W$ each) or two nodes in $W$ (and one in $U,V$).
    To handle these, we create $O(\log^2 n)$ subgraphs of $G_{bb'}$ where we also go through all $O(\log^2 n)$ choices of pairs of bits of $q$ and $r$ to make sure that $u+w\neq 0\pmod q$ and $u+v\neq 0\pmod r$. We obtain a poly-logarithmic number of subgraphs where we are guaranteed that if there were a \ThreeSum in $A$, in one of the subgraphs, this \ThreeSum is represented by a triangle and there are no $4$-cycles.    
    
    The number of vertices $N$ in each of the poly-logarithmic number of graphs is $\Theta(n^{4/3})$, and the maximum degree is $O(n^{1/3}) = O(N^{1/4})$, so the $C_{\le 4}$ detection instance requires $n^{2-o(1)} = N^{1.5-o(1)}$ time under the \StrongerThreeSum hypothesis.
\end{proof}

\ifC
\bibliography{articles}
\appendix
\else\fi

\section{Weighted girth approximation algorithm}\label{S-algo-full}
In this section, we prove the following theorem:
\Reminder{T-main-odd}

Let $k=2\alpha+1$.
Our algorithm is composed of an initialization phase and three different phases of girth approximations. Roughly speaking, Phase 1 handles the first ``half level'' using ideas from \cite{Ducoffe19}, and Phases 2,3 use ideas from \cite{kadria2023improved} for the rest of the levels.

In the initialization phase, we initialize all the data structures needed for the algorithm; this phase is similar to $\Initialize$  from \cite{kadria2023improved}. For completeness, we describe $\Initialize$ in \Cref{AP-In}, with one major difference: we create a hierarchy of vertex sets using non-uniform sampling. We have $A_0=V$, $A_{\alpha+1}=\emptyset$, $A_1 = \Sample(A_0,n^{-1/(2\alpha+1)})$, and \\$A_i = \Sample(A_{i-1},n^{-2/(2\alpha+1)})$, for every $2\le i \le \alpha$. The reason for these different probabilities is to create a first ``half level''. Using this ``half level'', the algorithm can address odd values of $k$. In this phase, the current girth estimation $\hat{g}$ is initialized to $\infty$, the current shortest cycle $\hat{C}$ is set to $\emptyset$, and both the distance hash table $d$ and the predecessor hash table $\pi$ (used to reconstruct shortest paths) are initialized as empty hash tables. In addition, the algorithm also computes $\delta(p_i(u),u)$ for every $0\le i \le \alpha$, and calls $\Preprocess(G)$, as required by \cite{kadria2023improved}, to efficiently implement $\ClusterOrCycle$.

In Phase 1, for every vertex $u \in V$, the algorithm computes the subgraph $G[B_0(u)]$ using \Cref{L-Create_B_S} and calls $\MinCycle(u,G[B_0(u)])$. If the length of a shortest cycle that goes through $u$ in $G[B_0(u)]$ is smaller than  $\hat{g}$, then the algorithm updates $\hat{C}$ and $\hat{g}$.

In Phase 2, for every vertex $u \in A_1$, the algorithm computes $\ClusterOrCycle(u)$.
If a cycle shorter than the current estimate $\hat{g}$  is detected by  $\ClusterOrCycle(u)$ in $\Cl(u)$, then the algorithm updates $\hat{C}$  and $\hat{g}$.

Finally, in phase 3, for every edge $(v, w) \in E$ and every $0 \le i \le \alpha$, the algorithm checks whether the edge $(v, w)$ closes a cycle in the shortest paths tree rooted at $u = p_i(v)$. If $(v,w)$ indeed closes a cycle and the length of this cycle is smaller than $\hat{g}$, then the algorithm updates $\hat{C}$ and $\hat{g}$.

The algorithm checks whether $(v, w)$ actually closes a cycle rather than merely retracing a path in the shortest paths tree rooted at $u$  by  
verifying that $\pi(u, v) \ne (w, v)$ and $\pi(u, w) \ne (v, w)$. 
If $\pi(u, v) \ne (w, v)$ and $\pi(u, w) \ne (v, w)$, a cycle is indeed formed, and the length of the cycle is bounded from above by $\hat{g}'=d(u,v)+\ell(v,w)+\delta(u,w)$. If $\hat{g}'<\hat{g}$, the algorithm updates $\hat{g}$ accordingly. In this case, we succinctly represent the discovered cycle $\hat{C}$ with the triplet $(u,v,w)$. 

At the end of the algorithm, to obtain a cycle from the triplet $(u,v,w)$, the algorithm computes $u'$, the Lowest Common Ancestor (LCA) of~$v$ and~$w$ in the shortest paths tree rooted at~$u$, and returns the simple cycle $P(u',v) \cup P(u',w) \cup (v,w)$.
The algorithm returns $\hat{g}$ as the estimated girth, along with $\hat{C}$ as the corresponding cycle.

\begin{algorithm2e}[t]
\caption{$\Cycle(G,\alpha)$}\label{A-Cycle}
\tcc{Initialization phase}
$\hat{g}\gets \infty$, $\hat{C}\gets \emptyset$, $d\gets\HashTable()$, $\pi\gets\HashTable()$\\
$A_0\gets V$ ; $A_{\alpha+1}=\emptyset$ ; $A_1 \gets \Sample(A_0,n^{-1/(2\alpha+1)})$ \\
\lFor{$i \in [2,\alpha]$}{
    $A_i\gets \Sample(A_{i-1},n^{-2/(2\alpha+1)})$
}
$\Preprocess(G)$\\
Compute $p_i(u)$ and $d(u,p_i(u))$ for every $\pair{u,i}\in \pair{V,[\alpha]}$\\

\tcc{Phase 1:}
\For{$u\in V$}{
    Compute $G[B_0(u)]$ using~\Cref{L-Create_B_S}\\
    $\pair{\hat{g}',\hat{C}'}\gets \MinCycle(u,G[B_0(u)])$\\
    \If{$\hat{g}'\le \hat{g}$}{
        $\hat{g}\gets \hat{g}'$ ; $\hat{C}\gets \hat{C}'$
    }
}

\tcc{Phase 2:}
\For{$u\in A_1$}{
    $\pair{\hat{g}',\hat{C}'}\gets \ClusterOrCycle(u)$\\
    \If{$\hat{g}'\le \hat{g}$}{
        $\hat{g}\gets \hat{g}'$ ; $\hat{C}\gets \hat{C}'$
    }
}

\tcc{Phase 3:}
\For {$(v,w)\in E$}{ \label{L-Alg-Loop2}
    \For{$i\gets 0$ {\bf to} $\alpha$}
        {
        $u\gets p_i(v)$ ;
        $\hat{g}' \gets d(u,v) + \ell(v,w) + d(u,w)$ \;
        \If {$\hat{g}' < \hat{g}$ {\bf and} $\pi(u,v) \neq (w,v)$ {\bf and} $\pi(u,w) \neq (v,w)$ }{ 
        \label{L-Cycle-If-pi}
            $\hat{g} \gets \hat{g}'$ ;
            $\hat{C}\gets (u,v,w)$ \;
        }
    }
}
\lIf {$\hat{C}$ is a triplet $(u,v,w)$} {$u'\gets LCA(u,v,w)$; $\hat{C}=P(u',v) \cup P(u',w) \cup (v,w)$}
\Return $\pair{\hat{g},\hat{C}}$
\end{algorithm2e}

Next, we bound the running time of the algorithm and show that $\Cycle(G,k)$ takes $\Ot(m + n^{1 + 2/k}) = \Ot(m + n^{1 + 2/(2\alpha + 1)})$ time, since $k = 2\alpha + 1$.
\begin{lemma}\label{L-alg-runtime}
    $\Cycle(G,k)$ takes $O((m+kn^{1+2/(2\alpha+1)})\log{n} + \alpha m)$ time.
\end{lemma}
\begin{proof}
    From \Cref{L-Init}, initializing the graph takes $O((m + kn)\log n)$ time.  
    Since $|A_1| = O(n^{1 - 1/(2\alpha+1)})$, it follows from \Cref{L-Create_B_S} that computing $G[B_0(u)]$ for every $u \in V$ takes $O(n^{1 + 2/(2\alpha+1)} \log n)$ time.
    Computing $\MinCycle(u, G[B_0(u)])$ requires \\$O(|E(G[B_0(u)])| \log n) = O(|B_0(u)|^2 \log n) = O(n^{2/(2\alpha+1)} \log n)$ time per vertex. Hence, computing $\MinCycle(u, G[B_0(u)])$ for all $u \in V$ takes a total of $O(n^{1 + 2/(2\alpha+1)} \log n)$ time.
    
    From \Cref{L-ClusterOrCycle}, computing $\ClusterOrCycle(u)$ takes $O(|\Cl(u)| \log n)$ time.  
    By \Cref{L-cluster-size}, we have $\mathbb{E}[\sum_{u \in A_1} |\Cl_V(u)|] = O(kn^{1 + 2/(2\alpha+1)})$. Therefore, computing $\ClusterOrCycle(u)$ for all $u \in A_1$ takes  
    \(
    O\left(\sum_{u \in A_1} |\Cl_V(u)| \log n \right) = O(kn^{1 + 2/(2\alpha+1)} \log n)
    \)
    time.
    
    In the final loop, for each edge $(v, w) \in E$, the algorithm iterates over $\alpha$ possible vertices, which costs $O(\alpha)=O(k)$ time per edge. Therefore, this loop takes $O(km)$ time.
    Overall, we get that the running time of the algorithm is $O((m + \alpha n^{1 + 2/(2\alpha+1)}) \log n + \alpha m)=\Ot(m+n^{1+2/k})$, as required.
\end{proof}

Next, we prove the following lemma, which is a main tool used in our correctness proof. 

\begin{lemma}\label{L-bound-h2}
Let $C$ be a shortest cycle; that is, $\ell(C)=g$. 
Let $u \in C$, and let $(x, y)$ be the farthest edge from $u$ in $C$, i.e., $(x,y)=\arg\max_{(x,y)\in C}(\delta(u,(x,y)))$.

Then either $\hat{g} \le 2g$ or $\min(\delta(A_2,x),\delta(A_2,y)) \le g$.
\end{lemma}
\begin{proof}
    Let $u \in C$, and let $(x, y)$ be the farthest edge from $u$ in $C$. Without loss of generality, assume that $\delta(u, x) \ge \delta(u, y)$.  
    If $C\subseteq G[B_0(u)]$, then we have\\ $\hat{g} \le \ell(\MinCycle(u,G[B_0(u)])) = g \le 2g$, and the claim holds. Therefore, for the remainder of the proof, we assume that $C \not\subseteq G[B_0(u)]$.
    
    Since $C \not\subseteq G[B_0(u)]$, there exists a vertex $w \in C$ such that $w \notin B_0(u)$. 
    From the definition of $B_0(u)$, we have $\delta(u, A_1) \le \delta(u, w)$. Since $x$ is the farthest vertex from $u$ in $C$, we get that:
    \begin{equation}\label{e-0}
        \delta(u,A_1) \le \delta(u,w) \le \delta(u,x)
    \end{equation}
    
    Since $C$ is a shortest cycle and $x\in C$, we have that $C=P(u,x)\cup P_2(u,x)$, and therefore $\ell(C)=\ell(P(u,x))+\ell(P_2(u,x))$. Thus,
    \begin{equation}\label{e-1}
        g=\ell(C)=\ell(P(u,x))+\ell(P_2(u,x))=\delta(u,x)+\ell(P_2(u,x)).
    \end{equation}

    Let $r = \delta(u, A_1) + \ell(P_2(u, x))$. We show that:
    \begin{cclaim}\label{C-C-in-G_r}
        $C \subseteq G_r(p_1(u))$
    \end{cclaim}
    \begin{proof}
        In order to prove that $C \subseteq G_r(p_1(u))$ we prove that for every $e\in C$ we have that $\delta(p_1(u),e)\le r$ and therefore $e\in G_r(p_1(u))$, as required.
        
        Since $C=P(u,x)\cup P_2(u,x)$ we have that either   $e\in P(u,x)$ or   $e\in P_2(u,x)$. If $e\in P(u,x)$ then we have that $\delta(u,e)\le \ell(P(u,x)) \le \ell(P_2(u,x))$, and therefore $\delta(p_1(u),e) \stactri \le \delta(A_1,u)+\ell(P_2(u,x))=r$, as required.\footnote{Throughout the paper, $\stactri\le $ is a step that follows from the triangle inequality.}
        If $e\in P_2(u,x)$ then we have that $\delta(u,e) \le \ell(P_2(u,x))$, and therefore $\delta(p_1(u),e) \stactri \le \delta(A_1,u) + \ell(P_2(u,x))$, as required.
    \end{proof}

    We divide the rest of the proof into the case that $C \subseteq \Cl(p_1(u))$, and the case that $C \not\subseteq \Cl(p_1(u))$.
    If $C \subseteq \Cl(p_1(u))$, then together with \Cref{C-C-in-G_r} we have that $C \subseteq G_r(p_1(u)) \cap \Cl(p_1(u))$, and therefore from \Cref{L-ClusterOrCycle} we have that:
    \begin{align*}
        \hat{g} &\le \ell(\ClusterOrCycle(p_1(u))) \le 2r=2(\delta(p_1(u),u) + \ell(P_2(u,x))) \\
        &\stackrel{\ref{e-0}}{\le} 2(\delta(u,x)+\ell(P_2(u,x))) \stackrel{\ref{e-1}}{\le} 2g,
    \end{align*}
    as required.

    If $C \not\subseteq \Cl(p_1(u))$, then we must have a cycle edge $(s,t)\in C$ such that $(s,t)\notin \Cl(p_1(u))$.
    If $(s,t)\neq (x,y)$ (recall that $(x, y)$ is the farthest edge from $u$ in $C$), then since $C$ is a shortest cycle we have that $(s,t)\subseteq P(u,t)$ or $(t,s)\subseteq P(u,s)$, and therefore $\delta(u,(s,t))=\delta(u,s)$ or $\delta(u,(s,t))=\delta(u,t)$. Wlog, we assume that $\delta(u,(s,t))=\delta(u,s)+\ell(s,t)=\delta(u,t)$.
    Since $(s,t)\notin \Cl(p_1(u))$, it follows from the definition of $\Cl(p_1(u))$ that
    \begin{equation}\label{e-A-2-t}
        \delta(A_2,t) \le \delta(p_1(u),s)+\ell(s,t) \stactri\le \delta(A_1,u)+\delta(u,s) + \ell(s,t) = \delta(A_1,u) + \delta(u,t)
    \end{equation}
    
    Since $(s,t)\neq (x,y)$ it follows that $(s,t) \in P(u,x) \cup P(u,y)$. (See \Cref{fig:c_6_1} for an illustration of this case.)
    If $(s, t) \in P(u, y)$, then $\delta(u, y) = \delta(u, t) + \delta(t, y)$. Therefore:
    \begin{align*}
        \delta(A_2, y) &\stactri\le \delta(y,t)+\delta(A_2,t) \stackrel{\ref{e-A-2-t}}\le \delta(y,t)+\delta(p_1(u),u) + \delta(u,t)
        \\&\stackrel{x\notin B_0(u)}\le \delta(u,t)+\delta(t,y)+\delta(u,x) \stackrel{t\in P(u,y)}\le \delta(u,y)+\delta(u,x) \le g/2+g/2=g,
    \end{align*}
    where the last inequality follows from the fact that $\delta(u, z) \le \frac{g}{2}$ for every $z \in C$, as required. 

    If $(s, t) = (y, x)$, then $(y, x) \notin \Cl(p_1(u))$, and by the definition of $\Cl(p_1(u))$ we have $\delta(A_2,x) \le \delta(p_1(u),y)+\ell(y,x)$. Therefore:
    \begin{align*}
        \delta(A_2,x) &\le \delta(p_1(u),y) + \ell(y,x) \stactri\le \delta(p_1(u),u)+\delta(u,y)+\ell(y,x) \\
        &\stackrel{\ref{e-0}}\le \delta(u,x)+\ell(P_2(u,x)) \stackrel{\ref{e-1}}{=} g,
    \end{align*}
    as required.
\end{proof}

Next, using \Cref{L-bound-h2} we bound the value of $\hat{g}$ from above by $\frac{2}{3}(2\alpha+1)g$.
Let $C$ be a shortest cycle in $G$, i.e. $\ell(C)=g$. Recall that $M(C)$ is the length of a longest edge in $C$. 
To prove that $\hat{g} \le \frac{2}{3}(2\alpha+1)g$, we consider the following two cases: the case that $M(C) \le g/3$ (\Cref{L-Corr-small}) and the case that $M(C) > g/3$  (\Cref{L-corr-big}).\footnote{We remark that the distinction between $M(C) \le g/3$ and $M(C) > g/3$ was also considered in the correction proof of \cite{kadria2023improved}.}
If $M(C) \le g/3$, then we show that in the first two loops of the algorithm, a cycle of length at most $\frac{2}{3}(2\alpha+1)g$ is found. Otherwise, if $M(C) > g/3$, we consider a longest edge $(u, u') \in C$ and show that either a short cycle is found during the first two loops, or that while considering the edge $(u,u')$ in the third loop of the algorithm a cycle of length at most $\frac{2}{3}(2\alpha+1)g$ is found.
We start with the case where $M(C) \le g/3$. 
\begin{lemma}\label{L-Corr-small}
    If $M(C) \le g/3$ then $\hat{g}\le \frac{2}{3}(2\alpha+1)g$
\end{lemma}
\begin{proof}
Let $h_i(C)=\min_{w\in A_i}(\delta(w,C))$. (See \Cref{fig:h_i} for an illustration of the $h_i$.)
To prove the lemma, we prove the following claim (\Cref{C-bound-hi}). For every $2 \leq i \leq \alpha$, either  
$h_{i+1}(C) \leq h_i(C) + \frac{2}{3}g$, or $\hat{g} \leq 2h_i(C) + \frac{4}{3}g$.

Using this claim, we prove the lemma as follows. From~\Cref{L-bound-h2}, we have that either $h_2(C) \le g$ or $\hat{g} \le 2g$. 
If $\hat{g} \le 2g$, then the lemma holds since the approximation guarantee is met. Therefore, for the remainder of the proof, we assume that $h_2(C) \le g$.
 
By applying \Cref{C-bound-hi} $(\alpha - 2)$ times, starting from $h_2(C)$, we conclude that either
\[
h_{\alpha+1}(C) \le h_2(C) + (\alpha - 1) \cdot \frac{2}{3}g
\quad \text{or} \quad
\hat{g} \le 2\left(h_2(C) + \frac{2}{3}(\alpha - 2)g\right) + \frac{4}{3}g.
\]  
The first inequality leads to a contradiction, since $A_{\alpha+1} = \emptyset$, and therefore $h_{\alpha+1}(C) = \infty$.  
Therefore, the second inequality must hold, and we have that  
\[
\hat{g} \le 2(h_2(C)+\frac{2}{3}(\alpha-2)g)+\frac{4}{3}g \stackrel{h_2(C)\le g}\le 2\left(g + \frac{2}{3}(\alpha - 2)g\right) + \frac{4}{3}g 
= \frac{2}{3}(2\alpha + 1)g,
\]
as required.
Therefore, to complete the proof of the lemma, it remains to prove the following claim.
\begin{cclaim}\label{C-bound-hi}
    Either $\hat{g} \le 2h_i(C)+4g/3$ or $h_{i+1}(C) \le h_i(C)+\frac{2}{3}g$
\end{cclaim}
\begin{proof}
    Let $w\in C$ such that $\delta(p_i(w),w)=h_i(C)$.
    % Let $w\in A_i$ such that $\delta(w,C)=h_i(C)$. 
    If there exists $x\in C$ such that $\delta(A_{i+1},x) \le h_i(C)+\frac{2}{3}g$ then $h_{i+1}(C) \le \delta(A_{i+1}, x) \le h_i(C)+\frac{2}{3}g$, and the claim holds. Therefore, for the rest of the proof, we assume that $\delta(A_{i+1}, x) > h_i(C)+2g/3$, for every vertex $x\in C$, and show that in this case $\hat{g} \le 2h_i(C)+4g/3$. 
    
    We will show that in this case $C\subseteq \Cl\left(p_i(w)) \cap G_{h_i(C)+2g/3}(p_i(w)\right)$, and therefore from\\ \Cref{L-ClusterOrCycle} we have that $\hat{g}\le \ell(\ClusterOrCycle(w)) \le 2(h_i(C)+2g/3)=2h_i(C)+4g/3$, as required.

    First, we show that $C\subseteq  G_{h_i(C)+2g/3}(p_i(w))$.
    From \Cref{lem:main} with $r=h_i(C)$, we have that $C\subseteq G_{\delta(p_i(w),w)+\frac{1}{2}(\ell(C)+M(C))}(p_i(w))$. 
    Since $\delta(p_i(w),w)=h_i(C)$, we have that \\$C\subseteq G_{h_i(C)+\frac{1}{2}(\ell(C)+M(C))}(p_i(w))$.
    Since $M(C)\le g/3$, we have that $\frac{1}{2}(\ell(C)+M(C)) \le 2g/3$ and therefore $C\subseteq G_{h_i(C)+2g/3}(p_i(w))$, as required.

    Next, we show that $C\subseteq \Cl(p_i(w))$. Let $(s,t)\in C$, we will show that $(s,t)\in \Cl(p_i(w))$.
    Since $C\subseteq G_{h_i(C)+2g/3}(p_i(w))$ we have that $\delta(p_i(w), (s,t))\le h_i(C)+2g/3$.
    Recall that we are in the case where for every vertex $x\in C$ we have $\delta(x, A_{i+1}) > h_i(C)+2g/3$. Therefore, we get that: 
    $\delta(p_i(w), (s,t))\le h_i(C)+2g/3 < \delta(t, A_{i+1})$. Therefore, from the definition of $\Cl(p_i(w))$ we get that $(s,t)\in \Cl(p_i(w))$, as required.
\end{proof}
\end{proof}

Next, we show that $\hat{g}\le \frac{2}{3}(2\alpha+1)g$ in the case where $M(C) > g/3$.
\begin{lemma}\label{L-corr-big}
    If $M(C) > g/3$ then $\hat{g}\le \frac{2}{3}(2\alpha+1)g$.
\end{lemma}
\begin{proof}
    Let $(u,u')\in C$ such that $\ell(u,u')=M(C)$. Let $h_i'(C)=\min(\delta(u, A_{i}), \delta(u', A_{i}))$ (see \Cref{fig:h_i} for an illustration). 
    The proof follows the same framework as \Cref{L-Corr-small}, with one crucial difference: instead of bounding $h_i(C)$, we bound $h_i'(C)$.
    To prove the lemma, we prove the following claim (\Cref{C-induction-heavy}):
    for every $2 \leq i \leq \alpha$, either  
    $h_{i+1}'(C) \leq h_i'(C) + \frac{2}{3}g$,  
    or $\hat{g} \leq 2h_i'(C) + \frac{4}{3}g$.
    
    % This is similar to \Cref{L-bound-h2}. 
    Using this claim, we prove the lemma as follows.
    First, notice that since $(u,u')$ is a longest edge in $C$, there exists $v\in C$ such that $\delta(v,u)\le g/2$ and $\delta(v,(u,u')) > g/2$, and therefore $(u,u')$ is the farthest edge from $v$ in $C$.
    Therefore, we can apply \Cref{L-bound-h2} on $v$ and $(u,u')$ and get that either $\delta(u,A_2) \le g$ or $\delta(u',A_2) \le g$ or $\hat{g} \le 2g$.
    If $\hat{g} \le 2g$, then the lemma holds. Therefore, for the remainder of the proof we assume that $\delta(u,A_2) \le g$ or $\delta(u',A_2) \le g$. From the definition of $h_2'(C)$ we get that $h_2'(C) \le g$. 
     
    By applying \Cref{C-induction-heavy} $(\alpha - 2)$ times, starting from $h_2'(C)$, we obtain that either  
    \[
    h_{\alpha+1}'(C) \le h_2'(C) + (\alpha - 1) \cdot \frac{2}{3}g
    \quad \text{or} \quad
    \hat{g} \le 2\left(h_2'(C) + \frac{2}{3}(\alpha - 2)g\right) + \frac{4}{3}g.
    \]  
    The first inequality leads to a contradiction, since $A_{\alpha+1} = \emptyset$, and therefore $h_{\alpha+1}'(C) = \infty$.  
    Therefore, the second inequality must hold, and we have that  
    \[
    \hat{g} \le 2(h_2'(C)+\frac{2}{3}(\alpha-2)g)+\frac{4}{3}g \stackrel{h_2'(C)\le g}\le 2\left(g + \frac{2}{3}(\alpha - 2)g\right) + \frac{4}{3}g 
    = \frac{2}{3}(2\alpha + 1)g,
    \]
    as required.
    Therefore, to complete the proof of the lemma, we only need to prove the following claim.
    \begin{cclaim}\label{C-induction-heavy}
        Either $h_{i+1}'(C)\le h_i'(C)+2g/3$ or $\hat{g} \le 2h_i'(C)+4g/3$
    \end{cclaim}
    \begin{proof}
        Assume, without loss of generality, that $\delta(u, A_i) \le \delta(u', A_i)$, and therefore $\delta(u,A_i)=h_i'(C)$.
        If $\min\{\delta(u, A_{i+1}),\delta(u', A_{i+1})\} \le h_i'(C)+ (g-M(C))$ then the claim holds due to our assumption that $M(C)\geq g/3$.
        Therefore, for the rest of the proof, we assume that 
        \begin{equation}\label{eq:uaplus}
        \delta(u', A_{i+1}) > h_i'(C)+ (g-M(C))\textrm{ and } \delta(u, A_{i+1}) > h_i'(C) + (g-M(C)).\end{equation}

        Next, we show that $u'\in \Cl_V(p_{i}(u))\cap V_{h_i'(C)+ (g-M(C))}(p_{i}(u))$. Since $\ell(C) = g$, $(u, u')\in C$ and $\ell(u,u')=M(C)$ we get that $\delta(u,u') = \min\{M(C), g-M(C)\}\le g-M(C)$. From the triangle inequality we get that $\delta(p_i(u),u')\le \delta(p_i(u),u)+g-M(C)=h_i'(C)+g-M(C)$. Therefore, $u'\in V_{h_i'(C)+ (g-M(C))}(p_{i}(u))$. Since $\delta(p_i(u),u') \le h_i'(C)+g-M(C) < \delta(u', A_{i+1})$, we get that $u'\in \CL_V(p_i(u))$. Hence, $u'\in \Cl(p_{i}(u))\cap G_{h_i'(C) + (g-M(C))}(p_{i}(u))$.

        Let $r$ be the smallest number such that $\Cl(p_{i}(u))\cap G_r(p_i(u))$ contains a cycle. If $r\leq h_i'(C) + (g-M(C))$ then by Lemma~\ref{L-shortest-path-cluster} $\ClusterOrCycle(p_{i}(u))$, when called, finds a cycle of length at most $2r\leq 2h_i'(C) + 2(g-M(C)) \le 2h_i'(C) + 4g/3$, as required. 
        Thus, we assume that $r > h_i'(C) + (g-M(C))$. 

        Since $r > h_i'(C) +  (g-M(C))$ it follows from ~\Cref{L-ClusterOrCycle} that $\ClusterOrCycle(p_{i}(u))$
        computes $d(p_i(u),u')=\delta(p_i(u),u')$ and a shortest paths tree rooted at $p_i(u)$ that contains $u'$. 

        Next, we show that when $\Cycle$ considers the edge $(u,u')$ 
        it holds that $\pi(p_i(u), u)\neq (u',u)$ and $\pi(p_i(u), u')\neq (u,u')$.
        
        % \avi{there are many typos in the proof:}
        \begin{subclaim}\label{sub-cond-true}
            Either
            $\pi(p_i(u), u)\neq (u',u)$ and $\pi(p_i(u), u')\neq (u,u')$ 
            or $\hat{g}\le 2h_i'(C)+4g/3$
        \end{subclaim}
        \begin{proof}
            We first show that $\pi(p_i(u), u)\neq (u',u)$. Assume for the sake of contradiction that \\ $\pi(p_i(u), u) = (u',u)$.  This implies that  $\delta(p_i(u), u')< \delta(p_i(u), u)$. 
            Since it always holds that $\delta(u', A_i)\leq \delta(p_i(u), u')$, we get that $\delta(u', A_i) < \delta(u, A_i)$, a contradiction to our assumption that $\delta(u, A_i) \le 
            \delta(u', A_i)$.
            
            We now show that $\pi(p_i(u), u')\neq (u,u')$.
            Assume, for the sake of contradiction, that \\ $\pi(p_i(u), u')= (u,u')$. This implies that 
            $\delta(p_i(u),(u,u'))\leq \delta(p_i(u), u')\leq h_i'(C) + (g-M(C))$, and hence 
            $(u,u')$ is in $G_{h_i'(C)+ (g-M(C))}(p_{i}(u))$. 
            
            Since $u'$ is in $\Cl(p_{i}(u))$ it follows from Lemma~\ref{L-shortest-path-cluster} that the shortest path between $p_i(u)$ and $u'$ is in $\Cl(p_{i}(u))$, thus its last edge $(u,u')$ is in $\Cl(p_{i}(u))$.

            Consider the path $C'(u,u')$ between $u$ and $u'$ that uses the edges of $C\setminus \{(u,u')\}$. Its length is $g-M(C)$. The concatenation of $P(p_i(u), u)$ with $C'(u,u')$ is a path between $p_i(u)$ and $u'$ of length at most $h_i'(C)+(g-M(C))$, and thus the distance between $p_i(u)$ and each of the edges $C\setminus \{(u,u')\}$ is at most $h_i'(C)+(g-M(C))$ which implies the edges of $C\setminus \{(u,u')\}$ are in $G_{h_i'(C)+(g-M(C))}(p_{i}(u))$.

            Let $(s,t) \in C'(u,u')$, and assume that when traversing $C'(u,u')$ from $u$ to $u'$, the vertex $s$ is encountered before $t$.
            % Let $(s,t)\in C'(u,u')$ and assume that when going from $u$ to $u'$ on $C'(u,u')$ we first encounter $s$.
            Let $C'(t,u')$ be the path from $t$ to $u'$ in $C$ avoiding the edge $(u,u')$, and let $C'(u,s)$ be the path from $u$ to $s$. Therefore, $\ell(C'(u,s))=\ell(C)-\ell(s,t)-\ell(C'(t,u'))=g-M(C)-\ell(C'(t,u'))$.
            
            Next we show that $(s,t)$ satisfies $\delta(p_i(u),s)+\ell(s,t)<\delta(t, A_{i+1})$, and thus is in $\Cl(p_{i}(u))$. 
            From the triangle inequality, we get that 
            $$\delta(p_i(u),s)\stactri\leq \delta(p_i(u),u)+\delta(u,s) \leq \delta(p_i(u),u)+g-M(C)-\ell(s,t) - \ell(C'(t,u'))$$ 
            By adding $\ell(s,t)$ to both sides we get that: 
            $$\delta(p_i(u),s)+\ell(s,t)\leq  \delta(p_i(u),u)+g-M(C) - \ell(C'(t,u'))$$
            Since $\delta(u', A_{i+1}) \leq \ell(C'(t,u')) + \delta(t, A_{i+1})$ we get that
            $$\delta(p_i(u),s)+\ell(s,t) \leq h_i'(C) + g-M(C)-\ell(C'(t,u')) \stackrel{\ref{eq:uaplus}}{<}  \delta(u', A_{i+1}) - \ell(C'(t,u')) \stactri{\leq} \delta(t, A_{i+1}),$$ 
            and therefore $(s,t)\in \Cl(p_{i}(u))$, as required. 
            
            Overall, we get that $C\subseteq G_{h_i'(C)+ (g-M(C))}(p_{i}(u)) \cap \Cl(p_i(u))$ and therefore from~\Cref{L-ClusterOrCycle} we have that $\hat{g}\le 2(h_i'(C)+ (g-M(C))) \le 2h_i'(C)+4g/3$, as required.
        \end{proof}

        From~\Cref{sub-cond-true}, we have that either the condition of~\Cref{L-Cycle-If-pi} holds or the lemma holds. If the condition of~\Cref{L-Cycle-If-pi} holds we get that:
        % Since $\delta(p_{i}(u), u') \stactri\le \delta(A_i,u) + (g-M(C))$, $\delta(p_{i}(u),u) = \delta(u, A_i) \le h_i'(u)$, and $\ell(u, u')=M(C)$ we get that:
        \begin{align*}
            \hat{g} &\le \delta(p_{i}(u),u) + \delta(p_{i}(u), u') + \ell(u,u') \\
            &\stactri\le 2\delta(p_{i}(u),u) + \delta(u,u') + \ell(u,u')
            \\&\stackrel{h_i'(C)\le \delta(p_{i}(u),u)}\le 2h_i'(u) + \delta(u,u') + \ell(u,u') \\
            &\stackrel{\delta(u,u')\le (g-M(C))}\le 2h_i'(u) + (g-M(C)) + \ell(u,u')
            \\&\stackrel{\ell(u,u')=M(C)} \le 2h_i'(u) + (g-M(C)) + M(C) = 2h_i'(u) + g,
        \end{align*}
        and algorithm $\Cycle$ finds a cycle of length at most $2h_i'(u) + g \leq 2h_i'(u) + 4g/3$, as required.
    \end{proof}
\end{proof}

\Cref{T-main-odd} follows from \Cref{L-Corr-small,L-corr-big,L-alg-runtime}.
\section{Proof of \Cref{C-Correctness-jin2023}}\label{S-LB-full}

    Recall the construction of \cite{jin2023removing}.
    Let $A\subseteq [\pm n^2]$ be the input for \ThreeSum. Let $p,q,r \in [2C_1n^{2/3}\log^2{n},4C_1n^{2/3}\log^2{n}]$ be three distinct primes, sampled uniformly at random. Let $A'=\{a\pmod{pqr} \mid a\in A\}$. Let $U$ (resp. $V$,$W$) be a vertex set identified by all numbers $[0,pqr)$ that are congruent to $0$ mod $p$ (resp. $q$,$r$). We add an edge between $u\in U$ and $v\in V$ if $v-u\pmod{pqr}\in A'$, between $v\in V$ and $w\in W$ if $w-v\pmod{pqr}\in A'$ and between $w\in W$ and $u\in U$ if $u-w\pmod {pqr}\in A'$. See \Cref{fig:h_G_uvw} for an illustration of this graph. %\input{figures/figure_G_UVW}

    \cite{jin2023removing} shows that the graph has a triangle if and only if there is a solution for \ThreeSum in $A$.
    \begin{claim}[\Cref{C-Correctness-jin2023}]
        The graph has a triangle if and only if there is a solution for \ThreeSum in $A$.
    \end{claim}
    \begin{proof}
        The if condition is trivial, since if this graph has a triangle $(u,v,w)$, then $(v-u),(w-v),(u-w)\in A'$ form a solution for \ThreeSum in $A'$, and since $pqr>3\cdot C_1\cdot n^2$ it is also a solution in $A$. 

    For the other direction, suppose $A$ has a \ThreeSum solution $x+y+z=0\pmod{pqr}$. Consider the vertices $u\in U$, $v\in V$, and $w\in W$ such that:
    \begin{table}[H]
        \centering
        \begin{tabular}{c c c}
             $u=0 \pmod{p}$ & $u=-x \pmod{q}$ & $u=z \pmod{r}$  \\
             $v=x \pmod{p}$ & $v=0 \pmod{q}$ & $v=-y \pmod{r}$  \\
             $w=-z \pmod{p}$ & $w=y \pmod{q}$ & $w=0 \pmod{r}$  
        \end{tabular}
    \end{table}
    Notice that $v-u=x\pmod{pqr}$, $w-v=y\pmod{pqr}$, and $u-w=z\pmod{pqr}$, so $(u,v,w)$ forms a triangle before we remove the high-degree vertices in the graph.
    
    By previous discussions, the number of neighbors of $u$ in $V$ is the number of $a \in A$
    where $u = -a \pmod{q}$, which implies $a - x = 0 \pmod{q}$. Note that $q$ is a uniformly random prime from $[2n^{2/3} \log^2 n, 4n^{2/3} \log^2 n]$, which contains $\Theta(n^{2/3} \log n)$ primes by the prime number theorem. Also, $a-x$
    has $\Theta(\log n)$ distinct prime factors, so the probability $a - x = 0 \pmod{q}$ is $O(n^{-2/3})$. Therefore, if we
    use $\text{deg}(s, T)$ to denote the number of neighbors of $s$ in set $T$ , then $E[\text{deg}(u, V )] \le O(n^{1/3})$. We can
    similarly bound $E[\text{deg}(u, W )]$, so $E[\text{deg}(u)] \le O(n^{1/3}) \le 10C \cdot n^{1/3}$ for some sufficiently large constant
    $C$. By Markov's inequality, $Pr[\text{deg}(u) > Cn^{1/3}] \le \frac{1}{10}$ , so we remove vertex $u$ with probability at most $\frac{1}{10}$.
    Similarly, we remove vertex $v$ and $w$ with probability at most $\frac{1}{10}$. By union bound, the triangle $(u, v, w)$ will
    remain in the final graph with probability at least $1-\frac{3}{10}=\frac{7}{10}$.
    We can repeat this reduction $O(\log n)$ times to boost the success probability. 

    \end{proof}

\section{Upper and lower bounds for the \AllEdgeTriangleOrListing problem}
\subsection{Upper bound}
In this section, we present an upper bound for the ($2k,t$)-\AllEdgeTriangleOrListing problem. We prove:
\begin{theorem}
    There is an $O(\min(m^{1+\frac{1}{k+1}},n^{1+\frac{2}{k}}) + tk)$-time algorithm for the\\ ($2k,t$)-\AllEdgeTriangleOrListing problem.
\end{theorem}

To obtain our algorithm for \AllEdgeTriangleOrListing, we build upon the \\\DegenerateOrCycle procedure of~\cite{kadria2022algorithmic}, extending it to list all triangles that involve low-degree vertices (implemented in the first for-loop of \Cref{A-AllEdgeTriangleOrListing}). For the remaining (degenerate) part of the graph, we combine techniques from~\cite{kadria2022algorithmic}  with a refined variant of the \BallOrCycle procedure, which we denote by \BallOrCycleListing (implemented in the second loop in \Cref{A-AllEdgeTriangleOrListing} and \Cref{A-BallOrCycleListing}).

Let $D = 1+\min(n^{1/k},2m^{1/(k+1)})$. 
The algorithm consists of two steps. First, for every vertex $a$ of degree at most $D$, in $O(\text{deg}(a)^2)$ time the algorithm finds all the triangles that contain $a$ and removes $a$ from the graph.
After removing low-degree vertices, the remaining graph has a minimum degree of at least $D$. The algorithm uses this fact and runs a modified version of $\codestyle{Cycle}$ presented by \cite{kadria2022algorithmic} to remove a set of vertices only after their triangles were already listed in the call to \BallOrCycleListing (or stops when $t$ cycles of length $2k$ are found).

Formally, the algorithm works as follows.
For every $a\in V$ let $d_a= \text{deg}(a)$, and let $S$ be the set of vertices whose degree is at most $D$.

While $S$ is not empty, the algorithm picks an arbitrary vertex $a$ from $S$, and removes $a$ from $S$. For every pair of vertices $u,v\in N(a)$, if $(u,v)$ is an edge,  the triangle $(a,u,v,a)$ is reported.
Next, for every vertex $b\in N(a)$,  $d_b$ is decreased by $1$ and if $d_b\le D$ then $b$ is added to $S$. Then, the algorithm removes $a$ and $E(a)$ from $G$ and proceeds to the next iteration.

After the first while loop ends, the algorithm proceeds to solve \AllEdgeTriangleOrListing for the remaining vertices as follows.
Let $N$ be the number of vertices that remained in $G$. While $G$ is not empty, the algorithm picks an 
arbitrary vertex $a\in V(G)$. 
The algorithm searches for the first $i$ such that $|\VV_{\le i}(a)|\leq N^{i/k}$, where $\VV_{\le i}(a)=\BallOrCycleListing(a, i,t)$.
Notice that $i \le k$, since otherwise $|V_{\le k}(a)|>N$, a contradiction.
The algorithm then removes $\VV_{\le (i-2)}(a)$ from $G$ and proceeds to the next iteration of the while loop.
Pseudo-code for the algorithm exists in \Cref{A-AllEdgeTriangleOrListing}.

\begin{algorithm2e}[t]
    \caption{\AllEdgeTriangleOrListing($G,k,t$)}\label{A-AllEdgeTriangleOrListing}
    Let $m'=1+\ceil{nD}$, where $D := 1+\min(n^{1/k},2m^{1/(k+1)})$\;
    % \lIf{$m> m'$}{Arbitrarily remove $m-m'$ edges from $G$}
    $\forall a\in V: d_a\gets \text{deg}(a)$\;
    $S\gets \{a \mid d_a\le D\}$\;
    \While{$S\neq \emptyset$} {
        Pick $a\in S$ and let $S\gets S\setminus a$\;
        For all $u,v\in N(a)$, if $(u,v)\in E$ then report the triangle $(a,u,v,a)$\; \label{l-go-over-neighbors}
        For all $b\in N(a)$ let $d_b\gets d_b-1$ and if $d_b\le D$ add $b$ to $S$\;
        Remove $a$ and all of its incident edges from $G$\;
    }
    
    $N \gets |V(G)|$\;
    \While{$G \neq \emptyset$} {
        Let $a\in V(G)$ be an arbitrary vertex\;
        \For{$i\gets 2$ to $k$}{
            $\VV_{\le i}(a) \gets \BallOrCycleListing(a, i, t)$\;
            \lIf{$\VV_{\le i}(a) = \codeNull$}{\Return}
            \If{$|\VV_{\le i}(a)|\le N^{i/k}$} {
                Remove $\VV_{\le (i-2)}(a)$ and all its incident edges from $G$\; \label{l-remove-v-i-2}
                \codestyle{break}
            }
        }
    }
\end{algorithm2e}

The~\BallOrCycleListing procedure gets as input a vertex $v$, a positive integer $R$, and a positive integer $t$. 
Let $d$ be an array storing the distances from $v$, and let $d[v]=0$. 
Let $\pi$ be an array of predecessors, and let $\pi[v]=\Null$.
Let $S$ be a set of vertices discovered so far, initially containing only $v$. Let $Q$ be a queue of vertices, initially containing only $v$.

The procedure \BallOrCycleListing works as follows: 
While $Q$ is not empty, the first element of $Q$ is extracted. Let $u$ be the extracted vertex. If $d[u]>R$ then the procedure terminates and returns $S$.
For each vertex $w\in N(u)$, the algorithm proceeds as follows. 
If $w\in S$, a cycle has been detected, and the algorithm adds the cycle to $\CC_{\le 2k}$, a global set of all the cycles of length at most $2k$, found so far. If $\CC_{\le 2k}$ contains at least $t$ cycles, then the procedure ends.
If $w\notin S$, then  $d[w]$ is set to $d[u]+1$, $\pi[w]$ is set to $u$, and $w$ is added to $S$  and $Q$.
Pseudo-code for \BallOrCycleListing exists in \Cref{A-BallOrCycleListing}.

\begin{algorithm2e}[t]
    \caption{\BallOrCycleListing($v, R, t$)}\label{A-BallOrCycleListing}
    \tcc{$\CC_{\le 2k}$ is a global list of found cycles}
    $d:=\codestyle{Array}(), d[v]\gets 0$, $\pi:=\codestyle{Array}(), \pi[v]\gets \Null$\;
    $S\gets \codestyle{Set}(\{v\})$, $Q\gets \codestyle{Queue}(\{v\})$\;
    \While{$Q \neq \emptyset$} {
        $u \gets Q.\codestyle{Extract}()$\;
        \lIf{$d[u] > R$}{
            \Return $S$}
        \ForEach{$w \in N(u)$}{
            \If{$w\in S$} {
                $\CC_{\le 2k} \gets \CC_{\le 2k} \cup \{v, u,w,\pi[w]\}$\; 
                \tcc{In $O(k)$ time all the vertices of the $\{v, u,w,\pi[w]\}$ cycle can be retrieved using the array $\pi$}
                \lIf {$|\CC_{\le 2k}| \ge t$} {\Return \codeNull}\label{l-stop-found-a-lot-cycles}
            }
            \Else{
                  $d[w] \leftarrow d[u] + 1$, $\pi[w] \leftarrow u$\;
                  $S\gets S\cup \{w\}$\;
                  $Q \gets Q.\codestyle{Insert}(w)$\;
            }
        }
    }
\end{algorithm2e}

Next, we state and prove the properties of \BallOrCycleListing.
\begin{lemma}\label{L-BallOrCycleListing}
    \BallOrCycleListing($v,R,t$) either returns $\VV_{\le R}(v)$ or stops after $\CC_{\le 2k}$ contains at least $t$ cycles. 
    The running time of \BallOrCycleListing($v,R,t$) is $O(|\VV_{\le R}(v)|+k|\CC'_{\le 2k}|)$, where $\CC'_{\le 2k}$ is the cycles added to the global set $\CC_{\le 2k}$ by \BallOrCycleListing($v,R,t$).
\end{lemma}
\begin{proof}
    First, we show that \BallOrCycleListing($v,R,t$) either returns $\VV_{\le R}(v)$ or stops after $\CC_{\le 2k}$ contains at least $t$ cycles. If the algorithm terminates in \Cref{l-stop-found-a-lot-cycles} then $\CC_{\le 2k}$ contains at least $t$ cycles. Otherwise, if the algorithm did not terminate in \Cref{l-stop-found-a-lot-cycles} then the procedure computes the first $R$   levels of the \codestyle{BFS} tree rooted at $v$. Therefore, from the correctness of $\codestyle{BFS}$, the returned set is indeed $V_{\le R}(v)$, as required.
    
    Since every edge iterated either finds a new vertex in $\VV_{\le R}(v)$, in $O(1)$ time, or adds a new cycle to $\CC_{\le 2k}$ in $O(k)$ time, we get that the total running time is $O(|\VV_{\le R}(v)|+k|\CC'_{\le 2k}|)$, as required.
\end{proof}

Next, we bound the running time of \AllEdgeTriangleOrListing.
\begin{lemma}
    \AllEdgeTriangleOrListing($G,k,t$) runs in $O(m+\min(m^{1+\frac{1}{k+1}}, n^{1+\frac{2}{k}})+tk)$-time.
\end{lemma}
\begin{proof}
    The cost of removing a vertex $v$ from $G$ is $O(\text{deg}(v))$, and since a vertex can be removed only once from the graph, the cost of removing all the vertices through the algorithm is $\sum_{v\in V}O(\text{deg}(v))=O(m)$. The cost of reporting a cycle in \BallOrCycleListing is $O(k)$, or $O(1)$ in the reporting of the first for loop, therefore the total cost for reporting the $t$ cycles is $O(kt)$.

    We now analyze the cost of the first while loop in terms of $m$.
    First, we consider the case that $m<nD$, and therefore $D= 1+\min(n^{1/k},2m^{1/(k+1)})=2m^{1/(k+1)}+1$.
    Iterating over all pairs $u,w\in N(a)$ to detect triangles takes $O((d_a)^2)$ time, and after this iteration, the algorithm removes $a$ along with its $d_a$ remaining edges.
    Since we have extracted $a$ from $S$, we have that $d_a\le D$, and therefore the amortized cost for each edge removal is $O(d_a^2/d_a)=O(d_a)=O(D)$. 
    Since at most $m$ edges are removed in this manner, we get that the cost of the first while loop is $O(mD)$. 
    
    Consider now the case that $m>nD$. This implies that $D=1+n^{1/k}$.
    Since we have extracted $a$ from $S$, we have that $d_a\le D$, and therefore iterating over $u,w\in N(a)$ takes $O(d_a^2)=O(D^2)$ time. Therefore, the total time to remove at most $n$ vertices is $O(nD^2)=O(n^{1+\frac{2}{k}})$, as required.
    
    Next, we consider the second while loop of the algorithm.
    Recall that we use $N$ to denote the number of vertices in $G$ after the first while loop ends. 
    
    We show that the running time of the second while loop is $O(N^{1+\frac2k})$. 
    Let $a$ be a vertex picked in some iteration of the second while loop. 
    First, observe that unless $|\CC_{\le 2k}| \ge t$, thereby satisfying the listing requirement, there must exist some iteration $i$ in which $|\VV_{\le i}(a)| \le N^{i/k}$.
    To see this, assume, for the sake of contradiction, that this is not the case, and in
    every iteration $|\VV_{\le i}(a)| > N^{i/k}$, for every $1 \le i \le k$. This implies that, in particular, $|\VV_{\le k}(a)|>N$, a contradiction to the number of vertices in $G$.
    
    Since $|\CC_{\le 2k}| < t$, from \Cref{L-BallOrCycleListing} we have that $\VV_{\le i}(a)$ contains all the vertices of distance at most $i$ from $a$, for every integer $i$ and $a\in V$,   called throughout the algorithm. 
    Let $2\le i\le k$ be the first iteration such that $|V_{\le i}(a)| \le N^{i/k}$. 
    Notice that $|V_{\le 0}(v)|=|\{v\}|=N^{0/k}$ and $|V_{\le 1}(v)|=|N(v)|\ge N^{1/k}$.
    Therefore, since $i$ is the first iteration such that $|V_{\le i}(a)| \le N^{i/k}$, 
    we get that   $|V_{\le (i-2)}(a)|\ge N^{(i-2)/k}$. 
    
   In $O(N^{i/k})$ time the algorithm removes $N^{(i-2)/k}$ vertices, therefore, the amortized time for each vertex removal is $O(N^{i/k}/N^{(i-2)/k})=O(N^\frac2k)$. Since at most $N$ vertices are removed we get that the running time for the second while loop is $O(N^{1+\frac2k})$, as required.

    In the case that $m\le nD$, then $D=1+2m^{1/(k+1)}$. Since the minimum degree in $G$ after the first while loop ends is $D$, we have $N\le m/D=\frac{m}{1+2m^{1/(k+1)}}=O(m^{1-\frac{1}{k+1}})$ vertices, and we get that the running time is 
    $$O(N^{1+\frac2k})=O((m^{1-\frac{1}{k+1}})^{1+\frac2k})=O(m^{\frac{k}{k+1}\cdot \frac{k+2}{k}})=O(m^{1+\frac{1}{k+1}}),$$
    as required.

    In the case that $m> nD$,
    we bound the running time with $O(n^{1+\frac2k})$, since $N\le n$.
\end{proof}
We now turn to prove the correctness of \AllEdgeTriangleOrListing.
\begin{lemma}
    \AllEdgeTriangleOrListing($G,k,t$) either finds $t$ cycles of length at most $2k$ or lists all the triangles in $G$.
\end{lemma}
\begin{proof}
    From \Cref{L-BallOrCycleListing} we have that the only case in which \BallOrCycleListing($v,R,t$) does not return $\VV_{\le R}(v)$ is when $|\CC_{\le 2k}| \ge t$. In this case, the lemma holds as the algorithm finds $t$ cycles of length at most $2k$. Therefore, throughout the proof, we assume that \BallOrCycleListing($v, R,t$) returns $\VV_{\le R}(v)$.
    
    To prove the correctness of the algorithm, we show that whenever a vertex $x$ is removed, all the triangles containing $x$ are reported. Consequently, since the algorithm terminates only when the graph is empty, it follows that all the triangles of the graph are listed.
    
    Let $x\in V$ be a vertex that was removed during the algorithm, and let $(x,u,w,x)$ be a triangle in $G$.
    
    If $x$ was removed in the first while loop, then in \Cref{l-go-over-neighbors} the algorithm considers every $v_1,v_2\in N(x)$, and in particular $u,w\in N(x)$ and reports the triangle $(x,u,w,x)$, as required.

    Otherwise, if $x$ was removed in \Cref{l-remove-v-i-2} during the $i$'th iteration, then we have that $x\in \VV_{\le (i-2)}(v)$, for some $v\in V$, and the algorithm computed \BallOrCycleListing($v,i,t$). 
    Since $x\in \VV_{\le (i-2)}(v)$ it follows that $\delta(x,v) \le i-2$, and since $u,w\in N(x)$, from the triangle inequality it follows that $\delta(v,w),\delta(v,u)\le \delta(v,x)+1 \le i-2+1=i-1$, and therefore $u,w\in \VV_{\le i-1}(v)$.
    
    During the call to \BallOrCycleListing($v,i,t$), all the edges of all the vertices $\VV_{\le i-1}$ were iterated, and in particular the edges $(x,u),(x,w)$, and $(u,w)$. Therefore, the algorithm reports the triangle $(x,u,w,x)$, as required.
\end{proof}

\subsection{Lower bound}
In this section, we consider the following \AllEdgeTriangleOrListing problem:
\begin{definition}[\AllEdgeTriangleOrListing]
    Let $G=(V,E)$, be an undirected graph.
    In the $(k, t)$-\AllEdgeTriangleOrListing problem, the goal is to return a solution to \textbf{one} of the following two problems:
    \begin{itemize}
        \item For every edge $e\in E$ determine if it is contained in a triangle and if so, return it.
        \item List \textbf{at least} $t$ cycles of length at most $k$.
    \end{itemize}
\end{definition}

To obtain our lower bound for \AllEdgeTriangleOrListing we use the following lemma of \cite{jin2023removing}, together with the sampling technique used by \cite{jin2023removing,DBLP:conf/stoc/AbboudBKZ22,DBLP:conf/stoc/AbboudBF23} and simple color coding introduced by \cite{AYZ97}.
\begin{lemma}[\cite{jin2023removing}]\label{C-1.6-jin}
    Under the \ThreeSum hypothesis, \AllEdgeTriangle on $n$-vertex graphs which have maximum degree at most $\sqrt{n}$ and have at most $n^{k/2}$ ~$k$-cycles for every $k \ge 3$ requires $n^{2-o(1)}$ time.
\end{lemma}

Next, we show a conditional lower bound of $m^{1+\frac{1}{k-1}-o(1)}+t$ for the \AllEdgeTriangleOrListing problem under the \ThreeSum Hypothesis, as stated in the following theorem.
\begin{theorem}\label{T-LB-m-alledgeorlisting}
    Under the \ThreeSum hypothesis, there is no 
    $O(m^{1+\frac{1}{k-1}-\eps}+t)$
    time algorithm for the $(k,t)$-\AllEdgeTriangleOrListing problem.
\end{theorem}
\begin{proof}
    We show that if there exists an algorithm $Alg$ that solves the\\ $(k,t)$-\AllEdgeTriangleOrListing problem in $O(m^{1+\frac{1}{k-1}-\eps}+t)$, for $\eps>0$,  time, then there 
    exists an algorithm that solves the \AllEdgeTriangle in graphs with maximum degree $O(\sqrt{n})$ and $O(n^{k/2})$ $k$-cycles in
    $O(n^{2-\delta})$, for $\delta>0$, time.
    By \Cref{C-1.6-jin} we get the desired conditional lower bound.
    
    Fix an instance of the \AllEdgeTriangle problem on an $n$-vertex graph $G=(V,E)$ with maximum degree $O(\sqrt{n})$ and $O(n^{k/2})$ $k$-cycles. Using the standard color-coding technique \cite{alon1995color}, we can assume that $G$ is tripartite with three vertex sets $A, B, C$.

    Let $0 < \sigma < \tfrac{1}{2}$ be a constant to be determined later. We partition the vertices of $G$ into $3n^{\sigma}$ sets, denoted $A_1,\dots,A_{n^{\sigma}},B_1,\dots,B_{n^\sigma},C_1,\dots,C_{n^\sigma}$, by picking each vertex of $A$ (resp. $B$ and $C$) independently and uniformly at random into one of the corresponding sets $A_i$ (resp. $B_i$ and $C_i$). It then suffices to solve the \AllEdgeTriangle problem on the sub-instances induced by $A_i \cup B_j \cup C_\ell$, for all $i,j,\ell \in [n^{\sigma}]$, which we denote by $G_{i,j,\ell}$.
    
    From Chernoff's inequality, it follows that, with high probability, 
    $G_{i,j,\ell}$ has $O(n^{1-\sigma})$ vertices, and maximum degree $O(n^{0.5-\sigma})$. Therefore, we get that with high probability:
    \begin{equation}\label{e-bound-E-G_ijl}
        |E(G_{i,j,\ell})| = O(n^{1-\sigma+0.5-\sigma})=O(n^{1.5-2\sigma})
    \end{equation}
    
    Since every cycle of length $k'$ remains in $G_{i,j,\ell}$ with probability $n^{-k'\sigma}$ we get that in expectation $G_{i,j,\ell}$ has at most $O(n^{k'/2-\sigma\cdot k'})=O(n^{k'(1/2-\sigma)})$ cycles of length $k'$ for any $3\le k'\le k$. 
    Since $\sigma<1/2$, the term $O(n^{k'(1/2-\sigma)})$ is  maximized when $k'=k$, and we get that the expected number of cycles of length between $3$ and $k$ is $O(n^{k(1/2-\sigma)})$. We denote this number with $C_{\le k}^{G_{i,j,\ell}}$.
    Let $C_{\le k}=\sum_{i,j,\ell\in [n^\sigma]}C_{\le k}^{G_{i,j,\ell}}$. 
    From the linearity of expectation we have that 
    $E[C_{\le k}]=O(n^{k(1/2-\sigma)+3\sigma})$.

    Therefore, by repeating this step $O(\log{n})$ times we get that with high probability:
    \begin{equation}\label{e-c-k-small}
        C_{\le k} = O(n^{k(1/2-\sigma)+3\sigma})
    \end{equation}

    Consider a specific subgraph $G_{i,j,\ell}$. 
    Notice, that solving\\ $(k,t)$-\AllEdgeTriangleOrListing, for $t>C_{\le k}^{G_{i,j,\ell}}$, in $G_{i,j,\ell}$ is equivalent to solving \AllEdgeTriangle in $G_{i,j,\ell}$. Using a geometric search in $O(\log{n})$ calls to $Alg$ we can find $t$ that satisfies $C_{\le k}^{G_{i,j,\ell}}< t\le 2C_{\le k}^{G_{i,j,\ell}}$.
    Therefore, using $Alg$ we get an $\Ot(|E(G_{i,j,\ell})|^{1+\frac{1}{k-1}}+t)$ time algorithm for \AllEdgeTriangle in $G_{i,j,\ell}$.
    % Therefore, using $Alg$ we solve \AllEdgeTriangle in $G_{i,j,\ell}$ in $\Ot(|E(G_{i,j,\ell})|^{1+\frac{1}{k-1}}+t)$ time.
    
    Next, we show that by using $Alg$ to solve \AllEdgeTriangle in all the $G_{i,j,\ell}$ instances we solve \AllEdgeTriangle in the original graph in $O(n^{2-\delta})$, for $\delta>0$, time which contradicts the \ThreeSum hypothesis.

    Recall that we assumed that the running time for $Alg$ is $O(m^{1+\frac{1}{k-1}-\eps}+t)$. Therefore,
    by using $Alg$ to solve $\AllEdgeTriangle$ in all instances $G_{i,j,\ell}$, for $i,j,\ell\in [n^\sigma]$, we get that the total running time is:
    \begin{align*}
        \sum_{i,j,\ell}\Ot(|E(G_{i,j,\ell})|^{1+\frac{1}{k-1}-\eps}+C_{\le k}^{G_{i,j,\ell}}) &\stackrel{\ref{e-bound-E-G_ijl}}{=}\Ot(n^{3\sigma}\cdot (n^{1.5-2\sigma})^{1+\frac{1}{k-1}-\eps}+\sum_{i,j,\ell}C_{\le k}^{G_{i,j,\ell}})
        \\&=\Ot(n^{3\sigma+(1.5-2\sigma)(1+\frac{1}{k-1}-\eps)}+C_{\le k})
        \\&\stackrel{\ref{e-c-k-small}}=\Ot(n^{1.5+\sigma+\frac{1.5-2\sigma}{k-1}-\eps(3/2-2\sigma)}+n^{k(1/2-\sigma)+3\sigma})
    \end{align*}

    By picking $\sigma=1/2-\frac{1}{2k-6}+\eps'$, for constant $\eps'>0$, we get that the total running time is 
    \begin{align*}
        &\Ot(n^{1.5+\sigma+\frac{1.5-2\sigma}{k-1}-\eps(3/2-2\sigma)}+n^{k(1/2-\sigma)+3\sigma}) \\
        % \\&=O(n^{1.5+(1/2-\frac{1}{2k-6}+\epsilon')+\frac{1.5-2(1/2-\frac{1}{2k-6}+\epsilon')}{k-1}-\eps(3/2-2(1/2-\frac{1}{2k-6}+\epsilon'))}+n^{k(1/2-(1/2-\frac{1}{2k-6}+\epsilon'))+3(1/2-\frac{1}{2k-6}+\epsilon')})
        &=\Ot(n^{2-\eps(1/2+\frac{1}{k-3}-2\eps')+\eps'(1-\frac{2}{k-1})}+n^{2-\eps'(k-3)})
    \end{align*}
    which is truly subquadratic for every  $0<\eps'<\eps/100$, as required.
\end{proof}

\ifC
\else
\bibliography{articles}
\pagebreak
\appendix
\fi

\section{Definitions}
We start with several problem definitions and known results that will be used to prove \Cref{T-AE-approximation-LB-1}. Let $C_{\le k}$ denote a cycle of length at most $k$. Let the set $[\pm n]$ be $\{i\mid  |i| < n\}$.

\begin{definition}[\ThreeSum]
The \ThreeSum problem asks whether, given a set $A$ of integers, there exist elements $a, b, c \in A$ such that $a + b + c = 0$.
\end{definition}

\begin{definition}[\ZeroTriangle] Given an $n$-node graph $G=(V,E,\wt)$ where $\wt:E\rightarrow [\pm n^{c}]$, for some constant $c>0$, decide whether there are three vertices $p, q, r$ such that $(p,q,r)$ is a triangle in $G$, and $\wt(p,q,r,p)=0$ 
\end{definition}

The following hypothesis serves as the foundation for many conditional lower bounds:

\begin{hypothesis}[\ThreeSum Hypothesis]
    In the word-RAM model with $O(\log n)$ bit words, there is no $O(n^{2-\eps})$, for $\eps>0$, time algorithm for \ThreeSum.
\end{hypothesis}

Sidon sets (also known as Golomb rulers) were extensively studied in the field of additive combinatorics (e.g. see the survey of~\cite{obryant2004completeannotatedbibliographywork}).

\begin{definition}[Sidon Set]
    A set of integers $A$ is a Sidon set if there are no distinct $a,b,c,d\in A$ such that $a+b=c+d$.
\end{definition}

Jin and Xu~\cite{jin2023removing} recently employed Sidon sets in the context of the \ThreeSum problem. In particular, they proved the following result:
\begin{lemma}[\cite{jin2023removing}]\label{L-Sidon-Exact}
    Under the \ThreeSum hypothesis, in the word-RAM model with $O(\log n)$ bit words, there is no
    $O(n^{3-\eps})$, for $\eps > 0$, time algorithm for \ZeroTriangle where $\{\wt(u,v)\mid (u,v)\in E\}$ is a Sidon Set.
\end{lemma}

Vassilevska W. and Xu~\cite{williams2020monochromatic} studied the following problem.
\begin{definition}[$(\rho,t)$-\AllEdgeTriangleListing \cite{williams2020monochromatic}] \label{def:AETD}
    Let $G=(A\cup B\cup C,E)$, be a tripartite graph with maximum degree $O(n^{1-\rho})$.
    In the $(\rho, t)$-\AllEdgeTriangleListing problem, the goal is to list for each edge $e \in E$, all triangles that contain $e$, if there are at most $t$ such triangles; otherwise, to list any $t$ distinct triangles containing $e$, assuming there are at least $t$ such triangles.
\end{definition}
We define $\rho$-\AllEdgeTriangle to be $(\rho,1)$-\AllEdgeTriangleListing.
Vassilevska W. and Xu \cite{williams2020monochromatic} showed that, under the \ThreeSum hypothesis, for every $\rho \ge 0.5$, there is no $O(n^{3 - 2\rho - \eps})$ time algorithm that solves $\rho$-\AllEdgeTriangle. 

Our goal is to show hardness for girth approximation. Therefore, in the following definition, we generalize \Cref{def:AETD} 
of the $(\rho,t)$-\AllEdgeTriangleListing problem by introducing an additional parameter $k$ and asking for every edge to detect a $C_{\le k}$. 
This can be viewed as a $k/3$-approximation to the $(\rho,t)$-\AllEdgeTriangleListing.

\begin{definition}[$(k,\rho,t)$-\AllEdgeCycleListing]
    Let $G=(A\cup B\cup C,E)$, be a tripartite graph with maximum degree $O(n^{1-\rho})$.
    In the $(k, \rho, t)$-\AllEdgeCycleListing problem, the goal is to list for each edge $e \in E$, all $C_{\le k}$ that contain $e$, if there are at most $t$ such $C_{\le k}$; otherwise, to list any $t$ distinct $C_{\le k}$ containing $e$, assuming there are at least $t$ such $C_{\le k}$.
\end{definition}

We define $(k,\rho)$-\AllEdgeCycle to be $(k,\rho,1)$-\AllEdgeCycleListing.

\section{Lower bounds for all edge shortest cycle approximation}
\subsection{$m^{1+\frac{1}{k-1}-o(1)}$ lower bound for $k/3$-stretch for all edge shortest cycle}
In this section, we show that using the techniques of \cite{jin2023removing} it is also possible to prove a lower bound for all edge shortest cycle. We use the following lemma:
\begin{lemma}[\cite{jin2023removing}]\label{L-jin-xu-6.1}
    Fix any constant $\sigma \in (0, 0.5)$, and any integer $k \ge 3$. Under the \ThreeSum hypothesis, it requires $n^{2-o(1)}$ time to solve $n^{3\sigma}$ instances of All-Edges Sparse Triangle on tripartite graphs with $O(n^{1-\sigma})$ vertices and maximum degree $O(n^{0.5-\sigma})$, such that the total number of cycles of length at most $k$ over all instances is $O(n^{k/2-(k-3)\sigma})$
\end{lemma}

Next, we prove the lower bound for AESC using the methods of \cite{jin2023removing}.
\begin{theorem}\label{T-AE-approximation-LB-k}\Copy{T-AE-approximation-LB-k}{
    Under the \ThreeSum hypothesis, given a graph $G=(V, E)$ with $n$ vertices and maximum degree $n^{1/(k-2)}$, there is no $O(n^{1+2/(k-2) - \eps})$, for $\eps>0$, time algorithm, for a $(k/3-\delta)$, for $\delta>0$, approximation for AESC. In terms of the number of edges $m$, the lower bound is $m^{1+1/(k-1)-o(1)}$.
    }
\end{theorem}
\begin{proof}
    Assume that such an algorithm $Alg$ exists for the sake of contradiction. Equivalently, say the running time of $Alg$ is $O(n^{1+ \frac{2}{k-2}-\epsilon'})$ for some $\epsilon'> 0$. We apply \Cref{L-jin-xu-6.1}, with parameter $k - 1$.

    For every instance, we only need to report for each edge whether it is in a triangle. 
    Since $Alg$ is guaranteed to return a $k/3-\delta$ approximation, $Alg$ finds for every edge that was in a triangle a cycle of length at most $k-1$.
    We will run $Alg$ in every instance, and for every edge that participated in a cycle of length at most $k-1$ we compute in $O(n^{0.5-\sigma})$ time whether it was contained in a triangle by going over all the neighbors of one of the endpoints of the edge.
    
    We get that the total running time is:
    $$O(n^{3\sigma}\cdot (n^{1-\sigma})^{1+2/(k-2)-\epsilon'} + n^{\frac{k-1}{2}-(k-4)\sigma}\cdot n^{0.5-\sigma})$$
    Setting $\sigma=1/2-1/(2k-6)+\epsilon/4<1/2$ gives a truly subquadratic running time, which contradicts the \ThreeSum hypothesis.
\end{proof}

\subsection{Alternative lower bound proof for the $(5/3-\epsilon)$ short cycle approximations}
Next, we prove the lower bound for $k/3-\delta$ stretch for all edge cycle in an alternative way. We show that under the \ThreeSum hypothesis, for $\rho \ge 2/3$ 
there is no $O(n^{3-2\rho-\eps})$ time algorithm for the $(4,\rho,O(1))$-\AllEdgeCycleListing. We prove:
\begin{theorem}\label{T-AE-approximation-LB-1}\Copy{T-AE-approximation-LB-1}{
    Under the \ThreeSum hypothesis, given a graph $G=(V, E)$ with $n$ vertices and maximum degree $n^{1/3}$, there is no $O(n^{5/3 - \eps})$-time algorithm, for a $(5/3-\eps)$-approximation for AESC. In terms of the number of edges $m$, the lower bound is $m^{5/4-o(1)}$.
    }
\end{theorem}

The reduction follows the same steps as the reduction of \cite{williams2020monochromatic}, however we show that if the weights of the input graph are a Sidon set then the number of $C_4$ is very small (\Cref{C-no-4-cycles}), and therefore even approximating a shortest cycle for each edge, that is, reporting either a triangle or a $C_4$ is hard.
\begin{theorem}\label{T-AllEdgeCycleListing-Hard}
    Let $\rho \ge 2/3$. Under the \ThreeSum hypothesis, there is no $O(n^{3-2\rho-\eps})$, for $\eps > 0$, time algorithm for $(4,\rho,O(1))$-\AllEdgeCycleListing.
\end{theorem}
\begin{proof}
    From \Cref{L-Sidon-Exact} we have that there is no $O(n^{3-\eps})$ time algorithm for \ZeroTriangle, such that $\{\wt(u,v) \mid (u,v)\in E\}$ is a Sidon set. 

    We show that if there exists an algorithm $Alg$ that solves \\ $(4,\rho,T)$-\AllEdgeCycleListing, for some constant $T>0$, in $O(n^{3-2\rho-\eps})$ time, 
    then there is an $O(n^{3-\eps'})$ time algorithm for \ZeroTriangle, where $\{\wt(u,v)\mid (u,v)\in E\}$ is a Sidon set.
    
    Let $G=(A\cup B\cup C,E,\wt)$, where $\wt:E\rightarrow [\pm n^M]$, for some constant $M\ge 1$, be a tripartite instance of \ZeroTriangle, such that $\{\wt(u,v)\mid (u,v)\in E\}$ is a Sidon set.

    We pick an arbitrary prime $p \in [100n^M, Dn^M\log{n}]$, for large enough constants $M, D$.
    After we determine $p$, we can regard all the weights of the graph as elements in
    the field $F_p$ by taking the weight of every edge modulo $p$. Since the weight of each triangle is in $[-3n^M, 3n^M]$ and
    $p\ge 100n^M$, the set of zero triangles with respect to the new weights stays the same.

    Next, we define the graph $G'=(A\cup B\cup C,E,\wt')$ as follows.
    Let $x\in F_p$ be a random element from $F_p$, and for every $v\in A\cup B\cup C$ let $y_v$ be an independently random element from $F_p$.
    Let $a\in A,b\in B$ and $c\in C$.
    We define the weight function $\wt'$ as follows: 
    \begin{itemize}
        \item $\wt'(a,b)=x\cdot \wt(a,b)+y_a-y_b$.
        \item $\wt'(b,c)=x\cdot \wt(b,c)+y_b-y_c$.
        \item $\wt'(c,a)=x\cdot \wt(c,a)+y_c-y_a$.
    \end{itemize}
    Notice that if $\ell(a,b,c,a)=0$ then $\wt'(a,b,c,a)=0$, and if $\wt'(a,b,c,a)=0$ then $\wt(a,b,c,a)=0$ as long as $x\neq 0$ (which happens w.h.p).

    To use $Alg$, we need to have that the maximum degree is $O(n^{1-\rho})$, therefore we create subgraphs $G'_{i,j,k}=(A\cup B\cup C, E'_{i,j,k}, \ell')$ of $G'$ as follows.
    We split $F_p$ into $n^\rho$ ranges: $L_0,\dots,L_{n^\rho-1}$, such that $|L_i|\le \ceil{pn^{-\rho}}$, and $L_i=\{i\cdot \ceil{pn^{-\rho}} + j' \mid 0 \le j' < \ceil{pn^{-\rho}} \text{ and } i\cdot \ceil{pn^{-\rho}} + j' \le p\}$.
    For every triplet $i,j,k\in [n^\rho]$, 
    the graph $G'_{i,j,k}$ is defined to be the subgraph such that $\wt'(a,b)\in L_i$, $\wt'(b,c)\in L_j$ and $\wt'(c,a)\in L_k$.
    
    Notice, that for every $i,j\in [n^{\rho}]$, there are $O(1)$ options of $k$ such that $G'_{i,j,k}$ may contain a zero triangle, and therefore it suffices to list all the triangles in these $O(n^{2\rho})$ subgraphs. See \Cref{fig:h_G_ijk} for an illustration of $G'_{i,j,k}$. \begin{figure}
    \centering
    \tikzset{every picture/.style={line width=0.75pt}} %set default line width to 0.75pt        
    \tikzset{every picture/.style={line width=0.75pt}} %set default line width to 0.75pt        

\begin{tikzpicture}[x=0.75pt,y=0.75pt,yscale=-1,xscale=1]

%Shape: Circle [id:dp36731749732929797] 
\draw   (385,262) .. controls (385,248.19) and (396.19,237) .. (410,237) .. controls (423.81,237) and (435,248.19) .. (435,262) .. controls (435,275.81) and (423.81,287) .. (410,287) .. controls (396.19,287) and (385,275.81) .. (385,262) -- cycle ;
%Shape: Circle [id:dp7802472557212401] 
\draw   (487,156) .. controls (487,142.19) and (498.19,131) .. (512,131) .. controls (525.81,131) and (537,142.19) .. (537,156) .. controls (537,169.81) and (525.81,181) .. (512,181) .. controls (498.19,181) and (487,169.81) .. (487,156) -- cycle ;
%Shape: Circle [id:dp15955977766117824] 
\draw   (601,260) .. controls (601,246.19) and (612.19,235) .. (626,235) .. controls (639.81,235) and (651,246.19) .. (651,260) .. controls (651,273.81) and (639.81,285) .. (626,285) .. controls (612.19,285) and (601,273.81) .. (601,260) -- cycle ;
%Straight Lines [id:da17381160900349224] 
\draw    (423,240) -- (492,172) ;
%Straight Lines [id:da6398787523541117] 
\draw    (435,262) -- (601,260) ;
%Straight Lines [id:da8611738460767853] 
\draw    (533,170) -- (611,239) ;
%Rounded Rect [id:dp20839291490566692] 
\draw   (206,103.2) .. controls (206,95.91) and (211.91,90) .. (219.2,90) -- (429.8,90) .. controls (437.09,90) and (443,95.91) .. (443,103.2) -- (443,142.8) .. controls (443,150.09) and (437.09,156) .. (429.8,156) -- (219.2,156) .. controls (211.91,156) and (206,150.09) .. (206,142.8) -- cycle ;

% Text Node
\draw (400,250) node [anchor=north west][inner sep=0.75pt]  [font=\LARGE] [align=left] {$\displaystyle A$};
% Text Node
\draw (500,144) node [anchor=north west][inner sep=0.75pt]  [font=\LARGE] [align=left] {$\displaystyle B$};
% Text Node
\draw (615,248) node [anchor=north west][inner sep=0.75pt]  [font=\LARGE] [align=left] {$\displaystyle C$};
% Text Node
\draw (409.88,218.38) node [anchor=north west][inner sep=0.75pt]  [rotate=-317] [align=left] {$\displaystyle \ell '( a,b) \in L_{i}$};
% Text Node
\draw (466.96,266.16) node [anchor=north west][inner sep=0.75pt]  [rotate=-359.77] [align=left] {$\displaystyle \ell '( c,a) \in L_{k}$};
% Text Node
\draw (557.12,163.19) node [anchor=north west][inner sep=0.75pt]  [rotate=-41.27] [align=left] {$\displaystyle \ell '( b,c) \in L_{j}$};
% Text Node
\draw (213,102) node [anchor=north west][inner sep=0.75pt]  [font=\footnotesize] [align=left] {$\displaystyle x,\ y_{v} \ \leftarrow Random( F_{p})$\\$\displaystyle L_{i} \ =\left\{i\cdot \lceil pn^{-\rho } \rceil +j\mid 0\ \leq j< \lceil pn^{-\rho } \rceil \ \right\}$\\$\displaystyle \ell '( u,v) =x\cdot \ell ( u,v) +y_{u} -y_{v}$};

\end{tikzpicture}

    \caption{Illustration of $G'_{i,j,k}$.}
    \label{fig:h_G_ijk}
\end{figure}
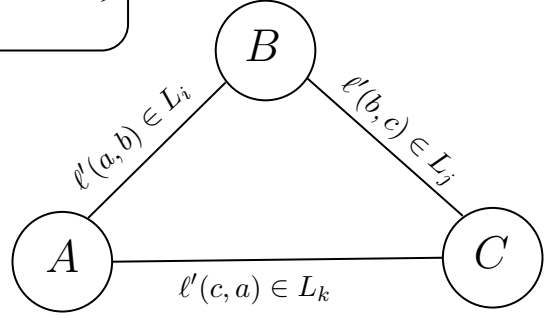

    We are going to solve the \ZeroTriangle problem as follows. 
    Let $G=(V,E,\ell)$ be the input to \ZeroTriangle, we create $G'_{i,j,k}$ as described above, and for every $i,j,k$ such that $G'_{i,j,k}$ may contain a zero triangle, we call $Alg$ to list at most $T$ distinct $C_{\le 4}$ for every edge. Then we go over each reported cycle and check if it is a $0$ weight triangle, if so we return it, otherwise we proceed. If we did not find a $0$ weight triangle in any $G'_{i,j,k}$ we return that there is no $0$ weight triangle in $G$.

    Since we assume $Alg$ takes $O(n^{3-2\rho-\eps})$ time, we get that the running time of the algorithm for \ZeroTriangle is $O(n^{2\rho}\cdot n^{3-2\rho-\eps})=O(n^{3-\eps})$.
    From \Cref{L-Sidon-Exact} it follows that an $O(n^{3-\eps})$ time algorithm for $\ZeroTriangle$ whose weights are a Sidon set refutes the \ThreeSum hypothesis.

    Therefore, to complete the proof, it remains to show that the algorithm described above returns a $0$ weight triangle if and only if there is a $0$ weight triangle in $G$. 

    The if direction follows easily since if the algorithm returns a $0$ weight triangle $(a,b,c)$, then we have that $\ell'(a,b,c, a)=0$, and as described above $\ell(a,b,c, a)=0$, as required.

    The more interesting part is to show that if there exists a triangle $(a,b,c)\in G$ such that $\wt(a,b,c,a)=0$, then w.h.p the algorithm returns a $0$ weight triangle.
    Consider the graph $G'_{i,j,k}$, such that $(a,b)\in L_i, (b,c)\in L_j$ and $(c,a)\in L_k$. We show that, with constant probability, the edge $(a,b)$ participates in $<T$ distinct $C_{\le 4}$ cycles in $G'_{i,j,k}$, for some constant $T > 0$, excluding cycles that are $0$ weight triangles. Therefore $Alg$, which reports at least $T$ distinct $C_{\le 4}$ for the edge $(a,b)$, must report a $0$ weight triangle that contains $(a,b)$. (To obtain a high probability guarantee, we repeat the algorithm $O(\log{n})$ times.)
    
    Therefore, we get that to complete the proof it suffices to show that for the edge $(a,b)$ the number of distinct $C_{\le 4}$ that are not a $0$ weight triangle is at most $T$.
    We divide this proof into two claims, first we show that the expected number of $C_4$ for each edge is $O(1)$ (\Cref{C-no-4-cycles}), then we show that the expected number of non-zero weight triangles for each edge is $O(1)$ (\Cref{C-no-nonzero-triangles}). 
    Using these two claims, we get that the expected number of distinct $C_{\le 4}$ that are not a $0$ weight triangle is at most $O(1)+O(1)=T'$, for some constant $T'$. Let the value of $T$ be at least $10T'$, then we get from Markov's inequality that the probability that the number of distinct $C_{\le 4}$ that are not a $0$ weight triangle is more than $T$ is at most $1/10$, as required.
    
    \begin{cclaim}\label{C-no-4-cycles}
        Let $(a,b)\in E'_{i,j,k}$ such that for some $c$, $(a,b,c)$ is a zero triangle in $E'_{i,j,k}$. $\EE(|\{C_4 \mid (a,b)\in C_4\}|)=O(n^{2-3\rho})$
    \end{cclaim}
    \begin{proof}
        Let $c'\in C, b'\in B$ be vertices such that $(a,b,c',b')$ is a $C_4$ in $G$.
        (The other case of cycles of the form $(a,b,a',b')$ follows using similar methods).
        From the definition of $G'_{i,j,k}$, it follows that 
        $(a,b,c',b')\in G'_{i,j,k}$ iff $\wt'(b,c'),\wt'(b',c')\in L_j$ and $\wt'(a,b')\in L_i$. Recall that $(a,b,c)$ is a zero triangle in $G'_{i,j,k}$ so that $\wt'(a,b)\in L_i, \wt'(b,c)\in L_j$. From the definition of $L_i,L_j$ we have that:
        \begin{equation*}
            \begin{cases}
                \wt'(b,c)-\wt'(b,c')\in [\pm pn^{-\rho}] \\
                \wt'(b,c)-\wt'(b',c') \in [\pm pn^{-\rho}] \\
                \wt'(a,b)-\wt'(a,b') \in [\pm pn^{-\rho}] \\
            \end{cases}
        \end{equation*}
        From the definition of $\wt'$, we get:
        \begin{equation}\label{eq4:1}
            \begin{cases}
                x(\wt(b,c)-\wt(b,c'))-y_c+y_c' \in [\pm pn^{-\rho}] \\
                x(\wt(b,c')-\wt(b',c'))+y_b-y_b' \in [\pm pn^{-\rho}]\\
                x(\wt(a,b)-\wt(a,b'))-y_{b}+y_{b'} \in [\pm pn^{-\rho}]\\
            \end{cases}
        \end{equation}
        Each of the three values in \Cref{eq4:1} is uniformly random, and therefore the probability that each happens is $1/n^{\rho}=n^{-\rho}$; we will show that they are independent and therefore the probability of the three occurring together is $n^{-3\rho}$.
      The first is independent of the other two since it is the only one that is dependent on $y_c'$.     
        It remains to show that $x(\wt(b,c')-\wt(b',c'))+y_b-y_b'$ and $x(\wt(a,b)-\wt(a,b'))-y_{b}+y_{b'}$ are independent.
        
        The sum of the two values is $x(\wt(b,c')+\wt(a,b)-\wt(b',c')-\wt(a,b'))$, and since the weights are a Sidon-set we have that  $\wt(b,c')+\wt(a,b)\neq \wt(b',c')+\wt(a,b')$, and therefore $\wt(b,c')+\wt(a,b)-\wt(b',c')-\wt(a,b')\neq 0$. Since $x$ is uniformly random and $\wt(b,c')+\wt(a,b)-\wt(b',c')-\wt(a,b')\neq 0$ we have that $x(\wt(b,c')+\wt(a,b)-\wt(b',c')-\wt(a,b'))$ is uniformly random.
        In addition, the second term in \Cref{eq4:1}, namely $x(\wt(b,c') - \wt(b',c')) + y_b - y_b'$, is independent of the summation term, as it involves $y_b'$, which does not appear in the sum.
        
        Thus, we get that the three variables are independent, each holds with probability $n^{-\rho}$ and overall $Pr[(a,b,c',b')\in G'_{i,j,k}]=n^{-3\rho}$. 
        From the linearity of expectation, we get that:
        $$\EE[|\{C_4 \mid (a,b)\in C_4\}|]\le \sum_{\pair{c',b'}\in \pair{C,B}}Pr[(a,b,c',b')\in G'_{i,j,k}]\le n^2\cdot n^{-3\rho}=n^{2-3\rho},$$
        as required. 
    \end{proof}

    \begin{cclaim}[Claim 3.6 \cite{williams2020monochromatic}] \label{C-no-nonzero-triangles}
        Let $(a,b)\in E_{i,j,k}$ such that for some $c$, $(a,b,c)$ is a zero triangle in $E_{i,j,k}$. $\EE[|\{C_3 \mid (a,b)\in C_3 \text{ and } \wt(C_3)\neq0\}|]=O(n^{1-2\rho})$
    \end{cclaim}
    \begin{proof}
        Let $c'\in C$ be a vertex such that $(a,b,c')$ is a triangle in $G$ and $\ell(a,b,c',a)\neq 0$.
        From the definition of $G'_{i,j,k}$, it follows that 
        $(a,b,c',a)\in G'_{i,j,k}$ iff $\wt'(b,c')\in L_j$ and $\wt'(c',a)\in L_k$. From the definition of $L_j,L_k$ we have that: 
        \begin{equation*}
            \begin{cases}
                \wt'(a,c)-\wt'(a,c')\in [\pm pn^{-\rho}] \\ 
                \wt'(b,c)-\wt'(b,c')\in [\pm pn^{-\rho}]
            \end{cases}
        \end{equation*}
        From the definition of $\wt'$, we get:
        \begin{equation*}
            \begin{cases}
                x(\wt(c,a)-\wt(a,c'))+y_c-y_c'\in [\pm pn^{-\rho}] \\ 
                x(\wt(b,c)-\wt(b,c'))+y_c'-y_c\in [\pm pn^{-\rho}]
            \end{cases}
        \end{equation*}
        Each of the two values is uniform, and therefore the probability that each happens is $1/n^{\rho}=n^{-\rho}$; we will show that they are independent and therefore the probability of both of them to occur together is $n^{-2\rho}$.
        
        The sum of the two values is $x(\wt(a,c) + \wt(b,c)-\wt(b,c')-\wt(a,c'))$. Since $\ell(a,b,c,a)=0$ and $\ell(a,b,c',a)\neq0$, we have that $\wt(a,c) + \wt(b,c)\neq \wt(b,c')+\wt(a,c')$, and therefore $\wt(a,c) + \wt(b,c)-\wt(b,c')-\wt(a,c')\neq 0$. Since it is uniformly random, we get that
        $x(\wt(a,c) + \wt(b,c)-\wt(b,c')-\wt(a,c'))$ is uniform random. The value $x(\wt(a,c)-\wt(a,c'))+y_c-y_c'$ is independent of the sum since it has an independent variable $y_c'$. Therefore, we get that the two variables are independent, each holds with probability $n^{-\rho}$ overall $Pr[(a,b,c',a)\in G'_{i,j,k}]=1/n^{2\rho}$. From the linearity of expectation, we get that 
        $$\EE[|\{{}_a\triangle_b \mid \wt(_a\triangle_b){\neq0}\}|]=\sum_{c'\in C}Pr[(a,b,c',a)\in G'_{i,j,k}]=O(n^{1-2\rho}),$$
        as required.
    \end{proof}
\end{proof}

To complete the proof of \Cref{T-AE-approximation-LB-1} we show that  $(4,\rho)$-\AllEdgeCycle is equivalent to $(4,\rho, O(1))$-\AllEdgeCycleListing in the following lemma.
(Note that we don't need the lemma in full generality as we only need that for every edge either a triangle is found, or $O(1)$ $4$-cycles are listed, but we provide the lemma for completeness.)

\begin{lemma}\label{L-AllEdgeCycleListing-1=O(1)}
    If there is an $\Ot(n^{1+\alpha})$ time algorithm for $(4, \rho)$-\AllEdgeCycle then there is also an $\Ot(n^{1+\alpha})$ time algorithm for $(4, \rho, O(1))$-\AllEdgeCycleListing that succeeds whp.
\end{lemma}
\begin{proof}
    We assume that there is an $\Ot(n^{1+\alpha})$ time algorithm ($Alg$) for $(4, \rho)$-\AllEdgeCycle and show an $\Ot(n^{1+\alpha})$ time (randomized) algorithm for $(4, \rho, O(1))$-\AllEdgeCycleListing.
    
    Let $G=(V, E)$ be a graph with maximum degree $n^{1-\rho}$, and a constant $T>0$ be an input for $(4, \rho, O(1))$-\AllEdgeCycleListing. 
    Our reduction algorithm works as follows: 
    For every $(u,v)\in E$, let $S_{(u,v)}$ be the set of $4$ cycles found for the edge $(u,v)$. The algorithm repeats the following step $T\cdot 2^Tc\log {n}$, for some constant $c$,  times. Let $G'$ be a subgraph of $G$ where every edge remains with probability $1/2$, i.e., $G'=(V,\Sample(E,1/2))$. We use algorithm $Alg$ to find a single $C_{\le 4}$ cycle $C$ for each edge $(u,v)$ and add this cycle to $S_{(u,v)}$. In addition, if $C$ is a $4$-cycle $u-v-x-y-u$, we also check if $u-v-y-x-u$, $u-v-y-u$, and $u-v-x-u$ are cycles and if so, also add them to $S_{(u,v)}$.

    At the end of the algorithm, for each edge $(u,v)$, the algorithm returns the set $S_{(u,v)}$ if $|S_{(u,v)}| < T$, or the first $T$ elements of $S_{(u,v)}$ otherwise.

    Since the running time of $Alg$ is $O(n^{1+\alpha})$, the running time for our reduction algorithm is $\Ot(T\cdot 2^T c\log{n} \cdot n^{1+\alpha})=\Ot(n^{1+\alpha})$, as required.
    
    Therefore, to complete the proof, it remains to prove the correctness of the algorithm.     
    Let $(u,v)\in E$, let $\Delta$ be the number of $C_{\le 4}$ containing $(u,v)$, and assume that the algorithm is in a stage where $|S_{(u,v)}|<T,\Delta$, that is, at least $\min(T,\Delta)-|S_{(u,v)}|$ cycles still need to be added to $S_{(u,v)}$.

    Next, we show a lower bound for the probability that the following three conditions hold together:
    \begin{enumerate}
        \item $(u,v)\in G'$. \label{en:pr:1}
        \item For every $C\in S_{(u,v)}$, it holds that $C\not\subseteq G'$. \label{en:pr:2}
        \item There exists $C_{\le 4}$ that contains $(u,v)$ in $G'$. \label{en:pr:3}
    \end{enumerate}
    It is straightforward to see that if these three conditions hold together, then a call to $Alg$ finds a new cycle. Therefore, after $T$ such successes, we find $\min(T,\Delta)$ cycles containing $(u,v)$, as required.
    
    Since $G'=(V,\Sample(E,1/2))$, we get that $Pr[(u,v)\in G']=1/2$.

Consider $S_{(u,v)}$. It consists of triangles, $A=\{(u-v-c_1-u),\ldots,(u-v-c_i-u),\ldots\}$ and of $4$-cycles $B=\{(u-v-x_1-y_1-u),\ldots,(u-v-x_i-y_i-u),\ldots\}$.
Suppose that there is some $C_{\le 4}$ containing $(u,v)$, $C$ which is not in $S_{(u,v)}$.

(Case 1:) $C$ is a triangle $(u-v-c)$. We know that $c\neq c_i$ for any $c_i$ such that $(u-v-c_i-u)\in A$ and that $c\neq x_i, c\neq y_i$ for any of the $(u-v-x_i-y_i-u)\in B$.
   This is because when we added $(u-v-x_i-y_i-u)$ to $S_{(u,v)}$, we added any other cycle containing $(u,v)$ that uses a subset of the vertex set of $(u-v-x_i-y_i-u)$.

   Because of this, the probability that given that $(u,v)$ is in $G'$, no cycles from $S_{(u,v)}$ are in $G'$ but $C$ is in $G'$ is at least the probability that
   \begin{enumerate}
\item For all $i$, $(u,c_i)\notin G'$ (destroying the cycles in $A$),
\item and for all $j$, $(x_j,y_j)\notin G'$ (destroying the cycles in $B$), and
\item $(u,c)\in G'$ and $(c,v)\in G'$.
   \end{enumerate}

All these events are independent\footnote{Except for the pairs of $4$-cycles in $B$ of the form $(u-v-x_i-y_i-u),(u-v-y_i-x_i-u)$ which are correlated but their correlation is actually in our favor since they both get destroyed if edge $(x_i,y_i)$ is removed.}. 

Hence the probability that $C$ is kept and all the cycles in $S_{(u,v)}$ are destroyed in $G'$ is at least $(1/2)\cdot (1/2)\cdot (1/2)^T=\Omega(1/2^T)$.

(Case 2:) $C$ is a $4$-cycle $u-v-x-y-u$. Here, the edge $(x,y)$ is distinct from all edges $(x_i,y_i)$ for $4$-cycles in $B$ by construction. 
If neither $x$ nor $y$ is one of the triangle nodes $c_i$ of triangles in $A$, then an analogous argument to the triangle case shows that with $\Omega(1/2^T)$ probability $C$ appears in $G'$ while the other cycles of $S_{(u,v)}$ are destroyed.

The hardest case is when $x$ and $y$ are both midpoints of some triangles in $A$. Wlog, assume that the $4$-cycle is $u-v-c_2-c_1-u$.
Here, the probability that given that $(u,v)$ is in $G'$, no cycles from $S_{(u,v)}$ are in $G'$ but $C$ is in $G'$ is at least the probability that
   \begin{enumerate}
\item For all $i\geq 2$, $(u,c_i)\notin G'$ (destroying the cycles in $A$ except $(u-v-c_1)$),
\item $(v,c_1)\notin G'$ (destroying $(u-v-c_1)$),
\item and for all $j$, $(x_j,y_j)\notin G'$ (destroying the cycles in $B$), and
\item $(u,c_1)\in G'$ and $(c_1,c_2)\in G'$ and $(v,c_2)\in G'$.
   \end{enumerate}

These events are again independent since we are talking about different edges and hence with $\Omega(1/2^T)$ probability $C$ appears in $G'$ while the other cycles of $S_{(u,v)}$ are destroyed.

    Therefore, by repeating $2^T c\log{n}$ times, we get that with high probability, a new cycle is found for every $(u,v)$. Since we repeat this step $T\cdot 2^T c\log{n}$, we get that with high probability at least $T$ cycles are found per edge (if they exist), as required.
\end{proof}

Combining \Cref{T-AllEdgeCycleListing-Hard} and \Cref{L-AllEdgeCycleListing-1=O(1)}, we obtain the following hardness result for all edge shortest cycle approximation.
\begin{corollary}\label{C-AE-approximation-LB-1}
    Let $\rho \ge 2/3$. Under the \ThreeSum hypothesis, there is no $O(n^{3-2\rho-\eps})$, for $\eps > 0$, time algorithm for $(4,\rho)$-\AllEdgeCycle.
\end{corollary}

\Cref{T-AE-approximation-LB-1} follows from \Cref{C-AE-approximation-LB-1} with $\rho=2/3$.

\section{$\Initialize$ and $\ClusterOrCycle$ of \cite{kadria2023improved}}\label{AP-OldToold}
For completeness, in this section, we describe the initialization algorithm of \cite{kadria2023improved} and the $\ClusterOrCycle$ procedure. 
\subsection{$\Initialize$}\label{AP-In}
The initialization algorithm, $\Initialize(G,\alpha)$, see Algorithm~\ref{A-Initialize}, receives the input graph $G=(V,E,\ell)$ and the parameter~$\alpha\ge 1$. 
The algorithm starts by sampling the vertex hierarchy $V=A_0\supseteq A_1 \supseteq A_2 \supseteq \dots \supseteq A_{\alpha+1} = \emptyset$.

Next, it initializes two empty \emph{hash tables} $d$ and $\pi$ used to store the distances and shortest paths already computed by the algorithm. When the algorithm discovers a distance $\delta(u,v)$ between two vertices $u,v\in V$, it inserts the pair $(u,v)$ into the hash table~$d$ with key
$\delta(u,v)$. For brevity, we write this as $d(u,v)\gets \delta(u,v)$. When we want to check whether $\delta(u,v)$ was already computed, we search $(u,v)$ in the hash table~$d$. If $(u,v)$ is found we retrieve $\delta(u,v)$. For brevity, we interpret $d(u,v)$ as a search for $(u,v)$ in the hash table~$d$. The search returns $\delta(u,v)$ if $(u,v)$ is in the table, or $+\infty$, if $(u,v)$ is not in the table, i.e., $\delta(u,v)$ is not yet known to the algorithm. (We assume that $d(u,v)$ searches both $(u,v)$ and $(v,u)$, or more efficiently, that all pairs $(u,v)$ stored in the table satisfy $u<v$.)

The hash table $\pi$ is similarly used to represent the shortest paths already found by the algorithm. If $d(u,v)<\infty$, then $\pi(u,v)$ is the last edge on a shortest path from~$u$ to~$v$. Thus, if $d(u,v)<\infty$ and $\pi(u,v)=(w,v)$ then $d(u,v)=d(u,w)+\ell(w,v)$.

We assume that each operation on the hash tables $d$ and $\pi$ takes constant expected time, as this can be achieved using standard hashing techniques. 

\begin{algorithm2e}[t] 
    \DontPrintSemicolon
 	\caption{$\Initialize(G = (V,E,\ell), \alpha)$} \label{A-Initialize}\label{A-Clusters}
 	\BlankLine
    $A_0\gets V$ ; $A_{\alpha+1} \gets \emptyset$; $A_1\gets \Sample(A_0,n^{-1/\alpha})$ \;
    \BlankLine
    \For {$i\gets 2$ {\bf to} $\alpha$} { \label{L-initialize-Ai}
    	$A_i\gets \Sample(A_{i-1}$,$n^{-2/\alpha}$) \;
    }
    \BlankLine
    $d\gets\HashTable()$\qquad \tcp{\rm Used to store computed
    distances.} 
    $\pi\gets\HashTable()$\qquad \tcp{\rm Used to store computed
    shortest paths.} 
    \BlankLine
    \For{$i\gets 1$ {\bf to} $\alpha$}
    {
        $\Dijkstra(G,A_i)\quad$ \tcp{\rm Finds $\delta(u,A_i)$ and $p_i(u)$ for every $u\in V$.}
        \For{$u\in V$} {
        $d(p_i(u),u)\gets \delta(u,A_i)$ \;
        }
    }
    \BlankLine
    $\Preprocess(G)$ \;
\end{algorithm2e}

The algorithm then computes the distances $\delta(u,A_i)$, for every $u\in V$ and $0\le i<k$. This is easily done by adding an auxiliary vertex $s_i$, connecting it with $0$ length edges to all vertices of~$A_i$ and then running Dijkstra from~$s_i$, as done in \cite{DBLP:journals/jacm/ThorupZ05}. This also computes $p_i(u)=\arg\min_{v\in A_i} \delta(u,v)$ for every $u\in V$ and $0\le i<k$ and a corresponding shortest path from~$u$ to~$p_i(u)$. 

Finally, \Initialize\ calls \Preprocess\ that performs preprocessing operations on the adjacency lists of all vertices. This processing includes sorting each adjacency list in non-decreasing order of edge length, and for every $0\le i<k$ building a binary tree on the edges of the vertex, as explained in Section~\ref{sub-next}. The total cost of all these preprocessing operations is $O(m\log n)$. 

\begin{lemma}\label{L-Init}
$\Initialize(G,\alpha)$ takes $O((m+kn)\log n)$ time.
\end{lemma}

\begin{proof}
The $k$ calls to Dijkstra's algorithm take $O(k(m+n\log n))$ time. Preprocessing the adjacency lists takes $O(m\log n)$. 
\end{proof}
Throughout the paper, a graph $G$ is said to be \emph{initialized} if the procedure $\Initialize(G, k)$ has already been called on it.

\subsection{Algorithm \texorpdfstring{$\ClusterOrCycle$}{ClusterOrCycle}} \label{sec:tools}

In this section, we describe an algorithm $\ClusterOrCycle$ (Algorithm~\ref{A-ClusterOrCycle}) that assumes that the graph has been initialized, and receives as input a vertex $u\in V$. The algorithm either returns the cluster $\Cl(u)$ or a cycle, as stated in the following lemma:

\Reminder{L-ClusterOrCycle}
% \begin{lemma}\label{L-ClusterOrCycle} Let $G=(V,E,\ell)$ be an initialized weighted undirected graph and let $u\in V$. If $\Cl(u)$ is a tree, then $\ClusterOrCycle(u)$ finds $\Cl(u)$, the distance $\delta(u,v)$ for each $v\in \Cl(u)$, and a tree of shortest paths from~$u$ to all vertices of $\Cl(u)$. Otherwise, if $r>0$ is the smallest number such that $\Cl(u)\cap G_r(u)$ contains a cycle, then $\ClusterOrCycle(u)$ returns a description of a cycle in $\Cl(u)\cap G_r(u)$ whose length is at most~$2r$. Furthermore, it returns $\Cl(u)\cap G_{<r}(u)$ and a tree containing shortest paths from $u$ to all vertices of $\Cl(u)\cap G_{<r}(u)$. $\ClusterOrCycle(u)$ can be implemented in $O(|\Cl_V(u)|\log n)$ time.
% \end{lemma}

$\ClusterOrCycle(u)$ goes through the appropriate steps to construct the cluster $\Cl(u)$, but stops early whenever a cycle in $\Cl(u)$ is encountered. This ensures that 
\\$\ClusterOrCycle(u)$ can be implemented in time proportional to the number of vertices in $\Cl(u)$, and not to the number of edges in $\Cl(u)$, which would have been too expensive. It uses a modification of Spira's \cite{Spira73} single-source shortest paths algorithm. 

Spira's algorithm assumes that the edges incident on each vertex are sorted in non-decreasing order of length. It may be viewed as a lazy version of Dijkstra's \cite{Di59} algorithm. In certain cases, it may find distances to all vertices without examining all edges. (This is possible as the adjacency lists of all vertices are assumed to be sorted by length. A recent application of Spira's algorithm can be found in \cite{WiZw15}.) 

When Dijkstra's algorithm discovers the distance from the source~$u$ to a new vertex~$v$, it immediately relaxes all the outgoing edges $(v,w)$ of~$v$. Spira's algorithm only relaxes the first outgoing edge of~$v$. The heap~$Q$ used by Spira's algorithm contains edges rather than vertices. 
Relaxing an edge $(v,w)$ amounts to inserting it into~$Q$ with key $d(u,v)+\ell(v,w)$. 
The algorithm also maintains a set~$U$ of vertices whose distance from the source~$u$ has already been found. Initially $U=\{u\}$. In each iteration, Spira's algorithm extracts an edge $(v,w)$ of minimum key from the heap~$Q$. If $w\notin U$, it adds $w$ to~$U$ and sets $d(u,w)\gets d(u,v)+\ell(v,w)$ which is guaranteed to be the distance from $u$ to~$w$. It now relaxes the first edge of~$w$ and the next edge of~$v$, i.e., the edge following $(v,w)$ in the sorted adjacency list of~$v$, if there is such an edge. If $w\in U$, the algorithm simply relaxes the next edge of~$v$. When $U=V$, the algorithm stops, even if there are still edges left in the heap~$Q$ and even if some edges were not examined yet.

The correctness of Spira's algorithm follows easily from the correctness of Dijkstra's algorithm, or can be proved directly using the same ideas used to prove the correctness of Dijkstra's algorithm. 

Algorithm $\ClusterOrCycle(u)$, shown as Algorithm~\ref{A-ClusterOrCycle}, uses the following modification of Spira's algorithm. It starts constructing $\Cl(u)$. The set $\Cl(u)$ denotes the set of vertices of the cluster discovered so far. Initially $\Cl(u)=\{u\}$. When the first edge $(v,w)$ for which $w\in \Cl(u)$ is extracted from the heap~$Q$, the algorithm stops as a cycle in $\Cl(u)$ is discovered, and the algorithm returns the discovered cycle. 

A non-trivial complication arises from the fact that we want $\ClusterOrCycle(u)$ to only examine edges that belong to $\Cl(u)$. Furthermore, for a correct implementation of Spira's algorithm, we need to examine these edges in non-decreasing order of length.

For the high-level description of Algorithm $\ClusterOrCycle(u)$, we assume that we have a function $\Next(u,v)$ that given a vertex~$v$ already known to be in $\Cl(u)$ gives us the next incident edge~$(v,w)$ of~$v$ that leads to a vertex~$w$ also in~$\Cl(u)$, in non-decreasing order of length. If there is no such next edge, then $\Next(u,v)$ returns $\Null$. The implementation of $\Next(u,v)$ is described in Section~\ref{sub-next}. It is shown there that it can be implemented in $O(\log n)$ time.

$\ClusterOrCycle(u)$ uses $\Next(u,v)$ via a function $\RelaxNext(u,v)$, see Algorithm~\ref{A-RelaxNext}, that uses $\Next(u,v)$ to extract the next eligible edge~$e$, if there is any, and relax it, i.e., add $e$ to the heap~$Q$ with key $d(v)+\ell(e)$. 

% \begin{minipage}{5in}
\begin{algorithm2e}[t]
    \DontPrintSemicolon
    \caption{$\ClusterOrCycle(u)$}\label{A-ClusterOrCycle}
    % $i \gets a(u)$ \;
    $d(u,u)\gets 0$ \; $\pi(u,u)\gets \Null$ \\
    $\cl(u) \gets \{u\}$ \\
    $Q\gets \Heap()$ \\
    $\RelaxNext(u,u)$ \\
    \BlankLine
    \While{$Q\ne\emptyset$}
    {
        \BlankLine
        $(v,w)\gets Q.\ExtractMin()$ \;
        \BlankLine
        \If{$w\in \cl(u)$}
        {\Return $\langle \, (u,v,w)\,,\,d(u,v)+\ell(v,w)+d(u,w)\,\rangle$ \;}
        \BlankLine
        $d(u,w) \gets d(u,v)+\ell(v,w)$ \; $\pi(u,w)\gets (v,w)$ \; $\cl(u)\gets \cl(u)\cup\{w\}$ \;
        \BlankLine
        $\RelaxNext(u,v)$ \; 
        $\RelaxNext(u,w)$ \;
    }
    \Return $\cl(u)$ 
\end{algorithm2e}
% \end{minipage}

\begin{algorithm2e}[t]
    \DontPrintSemicolon
    \caption{$\RelaxNext(u,v)$}\label{A-RelaxNext}
    $e\gets \Next(u,v)$ \;
    \If{$e\ne \Null$}
    {
        $Q.\Insert(e,d(u,v)+\ell(e))$ \;
    }
\end{algorithm2e}

We end this section with a proof of Lemma~\ref{L-ClusterOrCycle}.

\begin{proof}[Proof of Lemma~\ref{L-ClusterOrCycle}]
$\ClusterOrCycle(u)$ starts running Spira's algorithm on the implicitly represented cluster graph $\Cl(u)$. The algorithm extracts the edges $(v,w)$ of $\Cl(u)$ from the heap~$Q$ in non-decreasing order of their key $d(u,v)+\ell(v,w)$. When the first edge $(v,w)$ reaching a vertex~$w$ of~$\Cl(u)$ is extracted from~$Q$, then $\delta(u,w)=\delta(u,v)+\ell(v,w)$. The distance $d(u,w)$ is set accordingly, and $w$ is added to~$\Cl(u)$, the set of vertices of the cluster discovered so far. If a second edge $(v',w)$ reaches the same vertex~$w$ is extracted from~$Q$, then a cycle is detected and returned.
If $\Cl(u)$ does not contain a cycle, then from the correctness of Spira's algorithm, the algorithm $\ClusterOrCycle$ returns $\Cl(u)$ as required.

Let $r>0$ be the smallest number, as in the statement of the lemma, such that $\Cl(u)\cap G_r(u)$ contains a cycle. As $\Cl(u)\cap G_{<r}(u)$ does not contain a cycle, $\ClusterOrCycle(u)$ finds distances and shortest paths to all vertices of $\Cl(u)\cap G_{<r}(u)$ before a second edge reaching a vertex is found. The algorithm then starts finding vertices of distance exactly~$r$ from~$u$. As $\Cl(u)\cap G_r(u)$ contains a cycle, at some stage a second edge reaching a vertex in $\Cl(u)\cap G_r(u)$ must be found, and Spira's algorithm is aborted. This edge clearly closes a cycle of length at most $2r$, which is returned by the algorithm, as required.

Spira's algorithm spends $O(\log n)$ time on each edge $(v,w)$ it considers. This includes the $O(\log n)$ time taken by $\Next(u,v)$ to return the edge, the $O(\log n)$ (or $O(1)$) time needed to insert the edge to the heap~$Q$, and the $O(\log n)$ time needed for extracting it from the heap. The size of the heap is always at most the number of vertices in $\Cl(u)$, i.e., the vertices of the cluster discovered so far. As long as no cycles are found, the number of edges examined by Spira's algorithm is at most $2|\Cl(u)|-1$: the number of edges extracted from~$Q$ is $|\Cl(u)|-1$ and the number of edges in~$Q$ is at most $|\Cl(u)|$. When a cycle is found, the total number of edges examined is at most $2|\Cl(u)|$. The total running time is therefore $O(|\Cl_V(u)|\log n)$, as claimed.
\end{proof}

\subsubsection{Examining cluster edges in non-decreasing order of length}\label{sub-next}
Recall that if $u\in A_i\setminus A_{i+1}$ then $\Cl(u)=(\Cl_V(u),\Cl_E(u))$, where
\begin{align*}
    \Cl_V(u) &\EQ \{ v\in V \mid \delta(u,v)<\delta(v,A_{i+1}) \} \;, \\
    \Cl_E(u) &\EQ \{ (v,w)\in E \mid \delta(u,v)+\ell(v,w)<\delta(w,A_{i+1}) \} \;.
\end{align*}

Algorithm $\Preprocess$, called by \Initialize, defines $k$ \emph{shifted lengths} as follows, $\ell_i(v,w) = \ell(v,w)-\delta(w,A_{i+1})$ for each edge  $(v,w)\in E$, for every $i\in [0,\alpha]$. Now, if $u\in A_i\setminus A_{i+1}$ and $v\in \Cl_V(u)$ then $(v,w)\in\Cl_E(u)$ if and only if $\ell_i(v,w)<-\delta(u,v)$. We want to iterate over the edges of~$v$ that satisfy this condition in increasing order of their original length.

We can solve our problem using ideas borrowed from the \emph{priority search tree} of McCreight \cite{McCreight85}. 
We sort the $n$ points according to their $x$-coordinate and put them at the leaves of a binary search tree. (For simplicity, we may assume that~$n$ is a power of~2.) Each node of the tree contains the minimum $y$-coordinate among all the items in its subtree. All these values can be easily computed in $O(n)$ time by letting the value of each vertex be the minimum of the values of its two children.

Given an upper bound $y_0$ we can now easily find the item $(x_j,y_j)$ with the minimum $x$-coordinate that satisfies $y_j<y_0$. First, we check if the minimum $y$-value of the root is less than $y_0$. If not, then there is no point in satisfying the condition. Then, starting at the root, we repeatedly go to the left child if its minimum $y$ value is less than~$y_0$, and to the right child otherwise. The first item can thus be found in $O(\log n)$ time. Similarly, given an item, we can easily find the next item in $O(\log n)$ time. Thus, the first~$k$ items in non-decreasing order of their $x$-coordinates can be found in $O(k\log n)$ time.

\end{document}